\documentclass{article}

\PassOptionsToPackage{numbers, sort&compress}{natbib}
\usepackage[final]{neurips}
\setcitestyle{square}
\usepackage[utf8]{inputenc} % allow utf-8 input
\usepackage[T1]{fontenc}    % use 8-bit T1 fonts
\usepackage{hyperref}       % hyperlinks
\usepackage{url}            % simple URL typesetting
\usepackage{booktabs}       % professional-quality tables
\usepackage{amsfonts}       % blackboard math symbols
\usepackage{nicefrac}       % compact symbols for 1/2, etc.
\usepackage{microtype}      % microtypography
\usepackage{xcolor}         % colors

\usepackage{amsmath,amsfonts,amssymb,mathrsfs}
\usepackage{amsthm,etoolbox,changepage}
\begingroup
\makeatletter
\@for\theoremstyle:=definition,openquestion,remark,plain\do{%
	\expandafter\g@addto@macro\csname th@\theoremstyle\endcsname{%
		\addtolength\thm@preskip\parskip
	}%
}
\endgroup

\theoremstyle{plain}
\newtheorem{theorem}{Theorem}

\newtheorem*{lemma*}{Lemma}

\theoremstyle{definition}
\newcounter{defnctr}
\newtheorem{definition}[defnctr]{Definition}

\newcounter{propctr}
\newtheorem{proposition}[propctr]{Proposition}

\newcounter{lemmactr}
\newtheorem{lemma}[lemmactr]{Lemma}

\theoremstyle{openquestion}

\theoremstyle{remark}

\newcommand{\namedref}[2]{\hyperref[#2]{#1~\ref*{#2}}}

\newcommand{\subfigureref}[2]{\hyperref[#1]{Figure~\ref*{#1}#2}}

\usepackage{hyperref}
\usepackage{xcolor}

\usepackage{algorithmicx}
\usepackage{algorithm}
\usepackage[noend]{algpseudocode}

\usepackage{tikz}
\usepackage{pgfplots}
\usetikzlibrary{positioning}
\usetikzlibrary{arrows.meta}
\usetikzlibrary{patterns}
\usepackage{wrapfig}
\usepackage{caption}
\usepackage{subfigure}
\usepackage{wrapfig}
\usepackage{comment}
\usepackage{optidef}

\usepackage{pdflscape}

\usepackage{booktabs}
\usepackage{numprint}

\usepackage{enumerate}

\usepackage{mleftright}
\usepackage{bbm} % indicator

\makeatletter
\newtheorem*{rep@theorem}{\rep@title}
\newcommand{\newreptheorem}[2]{%
\newenvironment{rep#1}[1]{%
 \def\rep@title{\normalfont \textbf{#2 \ref{##1}}} % Had to change this to avoid italicizing
 \begin{rep@theorem}}%
 {\end{rep@theorem}}}
\makeatother

\newreptheorem{theorem}{Theorem}
\newreptheorem{definition}{Definition}
\newreptheorem{proposition}{Proposition}

\newcounter{relctr} %% <- counter for relations
\everydisplay\expandafter{\the\everydisplay\setcounter{relctr}{0}} %% <- reset every eq
\newcommand\labrel[2]{%
  \begingroup
    \refstepcounter{relctr}%
    \stackrel{\textnormal{(\alph{relctr})}}{\mathstrut{#1}}%
    \originallabel{#2}%
  \endgroup
}
\AtBeginDocument{\let\originallabel\label} %% <- store original definition

\hypersetup{
    colorlinks,
    linkcolor={red!50!black},
    citecolor={blue!50!black},
    urlcolor={blue!80!black}
}

\title{A Theory of Universal Rate-Distortion-Classification Representations for Lossy Compression}

\author{
  Nam Nguyen \\
  Electrical and Computer Engineering \\
  Oregon State University \\
  \texttt{nguynam4@oregonstate.edu} \\
  \And
  Thinh Nguyen \\
  Electrical and Computer Engineering \\
  Oregon State University \\
  \texttt{thinhq@oregonstate.edu} \\
  \And
  Bella Bose \\
  Electrical and Computer Engineering \\
  Oregon State University \\
  \texttt{bella.bose@oregonstate.edu} \\
}

\begin{document}
\maketitle

\begin{abstract}  
In lossy compression, Blau and Michaeli \cite{blau2019rethinking} introduced the information rate-distortion-perception (RDP) function, extending traditional rate-distortion theory by incorporating perceptual quality. More recently, this framework was expanded by defining the rate-distortion-perception-classification (RDPC) function, integrating multi-task learning that jointly optimizes generative tasks such as perceptual quality and classification accuracy alongside reconstruction tasks \cite{Wang2024}. To that end, motivated by the concept of a universal RDP encoder introduced in \cite{UniversalRDPs}, we investigate universal representations that enable diverse distortion-classification tradeoffs through a single fixed encoder combined with multiple decoders. Specifically, theoretical analysis and numerical experiment demonstrate that for the Gaussian source under mean squared error (MSE) distortion, the entire distortion-classification tradeoff region can be achieved using one universal encoder. In addition, this paper characterizes achievable distortion-classification regions for fixed universal representations in general source distributions, identifying conditions that ensure minimal distortion penalty when reusing encoders across varying tradeoff points. Experimental results using MNIST and SVHN datasets validate our theoretical insights, showing that universal encoders can obtain distortion performance comparable to task-specific encoders, thus supporting the practicality and effectiveness of our proposed universal representations.
\end{abstract}

%===========================================================================================
\section{Introduction}
Rate-distortion theory has long served as the foundation for lossy compression, characterizing the minimum distortion achievable at a given bit rate~\cite{cover1999elements}. Conventional systems are typically evaluated using full-reference distortion metrics such as MSE, PSNR, SSIM, and MS-SSIM~\cite{wang2004image, wang2003multiscale}. However, recent research has demonstrated that minimizing distortion alone is insufficient to yield perceptually convincing reconstructions. This limitation is particularly evident in deep learning-based image compression, where empirical evidence suggests that improvements in perceptual quality often come at the expense of increased distortion~\cite{agustsson2019generative,blau2018perception}. 

To address this limitation, Blau and Michaeli~\cite{blau2019rethinking} introduced the rate-distortion-perception (RDP) framework, which incorporates perceptual quality, measured via distributional divergence, as an independent optimization axis. The RDP formulation reveals a fundamental tradeoff among rate, distortion fidelity, and perceptual realism. Practical implementations, particularly those using GANs~\cite{goodfellow2014generative}, have demonstrated high perceptual quality at low bit rates~\cite{arjovsky2017wasserstein, tschannen2018deep, nowozin2016f, larsen2016autoencoding}. Common no-reference perceptual metrics include FID~\cite{heusel2017gans}, NIQE, PIQE, and BRISQUE~\cite{mittal2011blind, mittal2012making, venkatanath2015blind}.

In the context of signal restoration, the authors in \cite{CDP} pioneered the extension of the perception-distortion tradeoff by introducing the classification-distortion-perception (CDP) framework. Specifically, they incorporated the classification error rate of the restored signal as a third dimension, complementing distortion and perceptual quality. Their work rigorously established the inherent tradeoff within the CDP framework, demonstrating that it is impossible to simultaneously minimize distortion, perceptual difference, and classification error rate.

Further extending this perspective, recent studies have integrated classification tasks into lossy compression frameworks, enabling multi-task optimization that bridges the gap between image compression and visual analysis \cite{Zhang2023, Wang2024}. Initially proposed by Zhang et al. \cite{Zhang2023}, the rate-distortion-classification (RDC) model established a unified framework for optimizing the trade-off among rate, distortion, and classification accuracy in lossy image compression. Through extensive statistical analysis of multi-distribution sources, the RDC model was shown to exhibit favorable theoretical properties, including monotonic non-increasing behavior and convexity under specific conditions. Building upon this foundation, Wang et al. \cite{Wang2024} formalized the concept further by introducing the rate-distortion-perception-classification (RDPC) function. Their results rigorously demonstrated inherent trade-offs within the RDPC framework, revealing that enhancements in classification accuracy generally incur higher distortion or reduced perceptual quality.

Given these intricate tradeoffs in lossy compression scenarios, a fundamental question emerges: are these tradeoffs inherently determined by the encoder's chosen representation, or can a single encoder representation adapt to multiple objectives via different decoding strategies? To investigate this, the concept of a universal RDP framework was introduced in \cite{UniversalRDPs}, utilizing a single fixed encoder in conjunction with multiple decoders to achieve diverse points within the distortion-perception space without retraining the encoder.

In this paper, motivated by this universal approach, we investigate universal representations that enable diverse distortion-classification tradeoffs through a single fixed encoder combined with multiple decoders. Specifically, theoretical analysis and numerical experiment demonstrate that for the Gaussian source under mean squared error (MSE) distortion, the entire distortion-classification tradeoff region can be achieved using one universal encoder. In addition, this paper characterizes achievable distortion-classification regions for fixed universal representations in general source distributions, identifying conditions that ensure minimal distortion penalty when reusing encoders across varying tradeoff points. Experimental results using MNIST and SVHN datasets validate our theoretical insights, showing that universal encoders can obtain distortion performance comparable to task-specific encoders, thus supporting the practicality and effectiveness of our proposed universal representations.

%===========================================================================================
\section{Related Work}
Lossy compression has traditionally been studied through the framework of rate-distortion theory~\cite{cover1999elements}, which characterizes the minimum bit rate required to encode a source within a specified distortion level. Classical information theory has also explored distribution-preserving lossy compression, aiming to maintain statistical properties of the source in the reconstruction~\cite{saldi2015output,saldi2013randomized,zamir2001natural}. Recent advances in generative modeling~\cite{huang2020evaluating} have reignited interest in these foundations, particularly in the context of machine learning and representation learning. Modern frameworks increasingly leverage rate-distortion principles to learn compact, information-constrained representations~\cite{brekelmans2019exact,alemi2018fixing,tschannen2018recent}.

Within the rate-distortion-perception framework, distortion is typically evaluated using full-reference metrics that compare the reconstructed signal with the original. These include mean squared error (MSE), peak signal-to-noise ratio (PSNR), and structural similarity index (SSIM)~\cite{wang2004image}. In contrast, perceptual quality is assessed using no-reference metrics, which rely on the statistical properties of the output alone. Examples include Fréchet Inception Distance (FID)~\cite{heusel2017gans}, PIQE~\cite{venkatanath2015blind}, NIQE~\cite{mittal2012making}, and BRISQUE~\cite{mittal2011blind}.

The emergence of generative adversarial networks (GANs) has significantly advanced the field of perceptual compression. GAN-based models can produce visually realistic reconstructions, motivating the use of trained discriminators as perceptual quality evaluators~\cite{larsen2016autoencoding}. Theoretically, this is supported by the correspondence between GAN objectives and statistical divergence measures~\cite{nowozin2016f,arjovsky2017wasserstein,mroueh2017mcgan}. Building on these insights, several works have integrated GAN-based regularization into compressive autoencoder frameworks~\cite{agustsson2019generative,blau2019rethinking}, enabling high perceptual quality even at extremely low bitrates~\cite{mentzer2020high}. 

Blau and Michaeli~\cite{blau2018perception} first formalized the perception-distortion tradeoff, which was later extended into the RDP framework~\cite{blau2019rethinking}. These contributions inspired a broader class of models for distribution-preserving lossy compression~\cite{tschannen2018deep}, underscoring the necessity of jointly optimizing for perceptual realism and distortion fidelity.

Further extending this perspective, the classification-distortion-perception framework~\cite{CDP} introduced classification accuracy as an additional optimization axis. This formulation revealed that perceptual quality, distortion, and classification performance are fundamentally at odds: improving one often degrades the others. In multi-task learning settings, recent work~\cite{Wang2024,Zhang2023} proposed the rate-distortion-classification and rate-distortion-perception-classification functions, formalizing the joint optimization of these competing objectives. Their analyses revealed structural properties such as convexity and monotonicity of the tradeoff regions and empirically confirmed that improving classification accuracy typically incurs higher distortion or reduced perceptual quality.

To alleviate the need for retraining encoders for every task-specific operating point, the concept of universal RDP representations was proposed in~\cite{UniversalRDPs}. This framework fixes the encoder and trains multiple decoders to support diverse perceptual-distortion tradeoffs. The authors proved that such universal representations are approximately optimal under certain conditions, offering practical advantages for multi-objective compression without sacrificing performance.

%===========================================================================================
\section{Rate-Distortion-Classification Representations}
Consider a source generating observable data $X \sim p_X$, which inherently contains multiple target labels represented by variables $S_1, \dots, S_K \sim p_{S_1, \dots, S_K}$. The observable data $X$ and these intrinsic variables are correlated, following a joint probability distribution $p_{X, S_1,\dots, S_K}$ over the space $\mathcal{X} \times \mathcal{S}_1 \times \dots \times \mathcal{S}_K$. While the target variables are not directly observable, they can be inferred from $X$.  For example, if \( X \) is an image, classification tasks can be object recognition or scene understanding.
\begin{figure}[h]
\centering
\vspace{-0.2cm}
\includegraphics[width=0.7\textwidth]{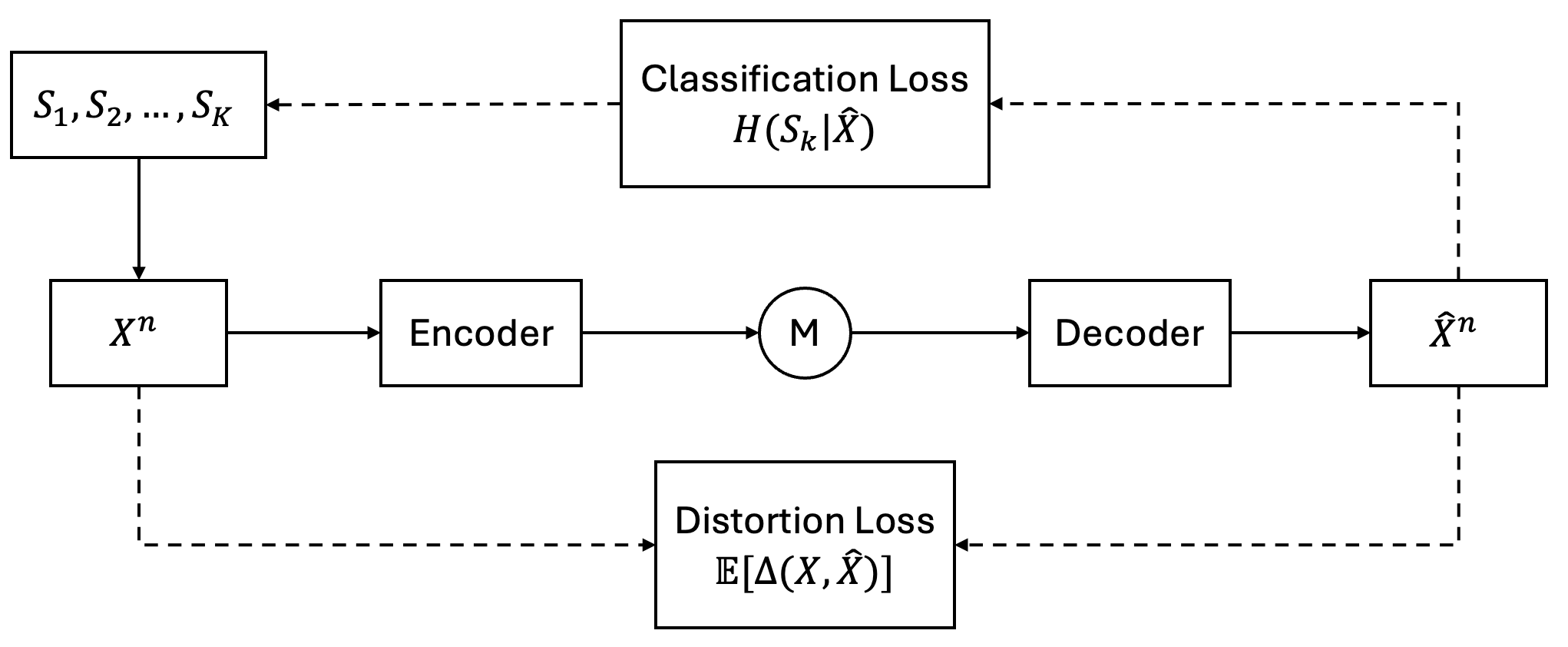}
\caption{Illustration of task-oriented lossy compression framework.}
\label{Lossy_Compression_Framework}
\end{figure}
As illustrated in Fig. \ref{Lossy_Compression_Framework}, the lossy compression process consists of an encoder and a decoder. Consider a source that generates an independent and identically distributed (i.i.d.) sequence $X_1, X_2, \dots, X_n \sim p_X$. The encoder, represented by the function $f: \mathcal{X}^n \to \{1,2,\dots,2^{nR}\}$, maps the input sequence $X^n$ into a compressed message $M$ at a rate of $R$ bits. This message is then processed by the decoder, defined as $g: \{1,2,\dots,2^{nR}\} \to \mathcal{\hat X}^n$, which reconstructs the sequence $\hat X^n$. The goal of this process is to preserve essential information while efficiently compressing data to meet the requirements of downstream tasks.

\textbf{Distortion constraint.} The reconstructed output $\hat{X}$ must satisfy the following distortion constraint:
\begin{equation}
    \mathbb{E}(\Delta(X,\hat{X})) \leq D,
\end{equation}
where $\Delta: \mathcal X \times \hat{\mathcal X} \to \mathbb{R}^+$ represents a distortion metric such as Hamming distortion or mean squared error. The expectation is computed over the joint distribution $p_{X, \hat X} = p_{\hat X|X} p_X$.

\textbf{Classification constraint.} We impose the following classification constraint:
\begin{equation}
    H(S_k|\hat{X}) \leq C_k, \qquad k \in [K],
    \label{ClassificationConstraint}
\end{equation}
for some $C_k > 0$. This ensures that the uncertainty of the classification variable $S_k$ given the reconstructed source $\hat{X}$ does not exceed $C_k$.

\textbf{Justification for the classification constraint.} To illustrate the relevance of constraint (\ref{ClassificationConstraint}), consider a single classification variable $S$. The system follows the Markov chain: $S \to X \to \hat{X} \to \hat{S}$. Applying Fano's inequality as:
\begin{equation}
    P(S \neq \hat{S}) \geq \frac{H(S | \hat{X}) - 1}{\log |S|},
\end{equation}
For a binary classification scenario ($|S| = 2$), this simplifies to: $P(S \neq \hat{S}) \geq H(S | \hat{X}) - 1$. Thus, instead of directly constraining the classification error probability $P(S \neq \hat{S})$, we impose a condition on its lower bound, namely $H(S | \hat{X})$. This justifies why the classification constraint $H(S | \hat{X}) \leq C$ serves as a meaningful and effective means to regulate classification error.

\textbf{Information RDC function.} To quantify the achievable rate under both distortion and classification constraints, the information rate-distortion-classification function for a source \( X \sim p_X \) is defined as follows:

\begin{definition}[Information Rate-Distortion-Classification Function]~\cite{Wang2024}
For a source \( X \sim p_X \) and classification variable \( S \), the \emph{information rate-distortion-classification function} is:
\begin{mini!}|s|[2] % mini! = minimize
{p_{\hat{X}|X}} % Optimization variable
{I(X;\hat{X})} % Objective function
{\label{RDC}} % Label for the optimization problem
{R(D,C) =} % Optimization result
\addConstraint{\mathbb{E}[\Delta(X, \hat{X})]}{\leq D}{\label{RDC-D}} % Distortion constraint
\addConstraint{H(S | \hat{X})}{\leq C.}{\label{RDC-C}} % Classification constraint
\end{mini!}
where $S$ represents the classification variable.
\end{definition}

%===========================================================================================
\section{Gaussian Source Case} \label{sec-RDC-Gaussian}
This section examines the rate-distortion-classification trade-off for a scalar Gaussian source. Leveraging the analytical tractability of the Gaussian setting, we derive closed-form expressions for both the RDC and DCR functions, offering valuable insights into the relationship between compression, distortion, and classification accuracy.

The closed-form expression of \( R(D,C) \) is given in the following theorem for a scalar Gaussian source.
\begin{figure}[!htbp]
\centering
\includegraphics[width=0.55\textwidth]{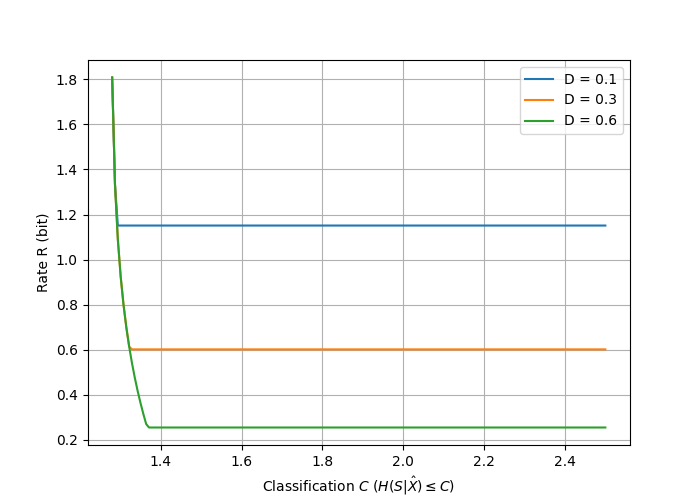}
\caption{Illustration of the information rate-distortion-classification function of a Gaussian source.}
\label{RDC}
\end{figure}
\begin{theorem}[Information Rate-Distortion-Classification Function for a Gaussian Source]\label{TheoremRDCGS}
~\cite{Wang2024} Let \( X\sim \mathcal{N}(\mu_X,\sigma_X^2) \) be a Gaussian source and \( S\sim \mathcal{N}(\mu_S,\sigma_S^2) \) be an associated classification variable, with a covariance of \( \text{Cov}(X,S) = \theta_1 \). The problem is feasible if $C \geq \frac{1}{2} \log\left(1 - \frac{\theta_1^2}{\sigma_S^2 \sigma_X^2}\right) + h(S)$. For the mean squared error distortion (i.e., \( \mathbb{E}[\Delta(X, \hat{X})] = \mathbb{E}[(X-\hat{X})^2] \)), the information rate-distortion-classification function is achieved by a jointly Gaussian estimator \( \hat{X} \) and is given by
  \begin{align*}
    R(D,C) =
    \begin{cases}
        \frac{1}{2} \log \frac{\sigma_X^2}{D}, \quad
        & D \leq \sigma_X^2 \left( 1 - \frac{1}{\rho^2} \left( 1 - e^{-2h(S) + 2C} \right) \right), \\
        -\frac{1}{2} \log \left(1 - \frac{1}{\rho^2} \left( 1 - e^{-2h(S) + 2C} \right) \right) \quad
        & D > \sigma_X^2 \left( 1 - \frac{1}{\rho^2} \left( 1 - e^{-2h(S) + 2C} \right) \right) \\
        0, \quad 
        & C > h(S) \text{ and } D > \sigma_X^2.
    \end{cases}
  \end{align*}
where \( \rho = \frac{\theta_1}{\sigma_S \sigma_X} \) represents the correlation coefficient between \( X \) and \( S \), while \( h(\cdot) \) denotes the differential entropy of a continuous random variable.
\end{theorem}
\begin{proof}
For completeness, we provide a detailed proof in Appendix~\ref{Appendix_Proof_RDC_GS}, extending the outline presented in~\cite{Wang2024}. This inclusion ensures the paper remains self-contained.
\end{proof}

The visualization of Theorem \ref{TheoremRDCGS} is provided in Figure \ref{RDC}, with parameters set to \( \sigma_X^2 = \sigma_S^2 = 1.0 \) and \( \theta_1 = 0.5 \). The trade-off between rate, distortion, and classification is clearly illustrated. 

This characterization highlights the intricate tradeoff between reconstruction fidelity and classification performance. In the first case of the theorem, the RDC function reduces to the classical rate-distortion formulation, with the classification constraint inactive. As the constraint tightens (i.e., \( C \) decreases), a higher rate is required to jointly satisfy both distortion and classification objectives. This dependency is reflected in the curvature of the RDC boundary illustrated in Fig.~\ref{RDC}.

Furthermore, we characterize the minimum achievable distortion as a function of \( C \) and \( R \) by the following theorem.

\begin{theorem}[Information Distortion-Classification-Rate Function for a Gaussian Source]\label{TheoremDCR_GS}
Consider a Gaussian source \( X\sim \mathcal{N}(\mu_X,\sigma_X^2) \) and an associated classification variable \( S\sim \mathcal{N}(\mu_S,\sigma_S^2) \) with covariance \( \text{Cov}(X,S) = \theta_1 \). The problem is feasible if the classification loss satisfies $ C \geq \frac{1}{2} \log\left(1 - \frac{\theta_1^2}{\sigma_S^2 \sigma_X^2}\right) + h(S)$. Under the mean squared error distortion, the information distortion-classification-rate function is given by
\begin{align*}
    D(C, R) = 
\begin{cases} 
    \sigma_X^2 e^{-2R}, \\ 
    \hspace{2cm} C \geq \frac{1}{2} \log\left(1 - \frac{\theta_1^2 (\sigma_X^2 - \sigma_X^2 e^{-2R})}{\sigma_S^2 \sigma_X^4} \right) + h(S) \\\\
    \sigma_X^2 + \frac{\sigma_S^2 \sigma_X^4 (1 - e^{2C - 2h(S)})}{\theta_1^2} 
    - 2 \frac{\sigma_S \sigma_X^3 \sqrt{(1 - e^{-2h(S) + 2C})(1 - e^{-2R})}}{\theta_1}, \\ 
    \hspace{2cm} \frac{1}{2} \log\left(1 - \frac{\theta_1^2}{\sigma_S^2 \sigma_X^2}\right) + h(S) 
    \leq C < \frac{1}{2} \log\left(1 - \frac{\theta_1^2 (\sigma_X^2 - \sigma_X^2 e^{-2R})}{\sigma_S^2 \sigma_X^4} \right) + h(S)\\
    \sigma_X^2, \\
    \hspace{2cm} C > h(S).
\end{cases}
\end{align*}
\end{theorem}
\begin{proof} 
The proof is provided in Appendix \ref{Appendix_Proof_DCR_GS}.
\end{proof}

For any fixed \( R \), as \( C \) increases from \( \frac{1}{2} \log\left(1 - \frac{\theta_1^2}{\sigma_S^2 \sigma_X^2}\right) + h(S) \) to \( \frac{1}{2}\log\left(1 - \frac{\theta_1^2 (\sigma_X^2 - \sigma_X^2 e^{-2 R})}{\sigma_S^2\sigma_X^4} \right) +h(S) \), the distortion \( D(C,R) \) decreases monotonically from 
\[
\sigma_X^2 + \frac{\sigma_S^2\sigma_X^4(1-e^{2C-2h(S)})}{\theta_1^2} - 2 \frac{\sigma_S \sigma_X^3 \sqrt{(1 - e^{-2 h(S) + 2C})(1 - e^{-2 R})}}{\theta_1}
\]
to the optimal value \( \sigma^2_X e^{-2R} \). Increasing \( C \) beyond this point does not yield further improvements in distortion.
\begin{figure}[h] 
    \centering
    \subfigure[$D(C,R)$ function for a Gaussian source.]{% 
        \includegraphics[width=0.48\textwidth]{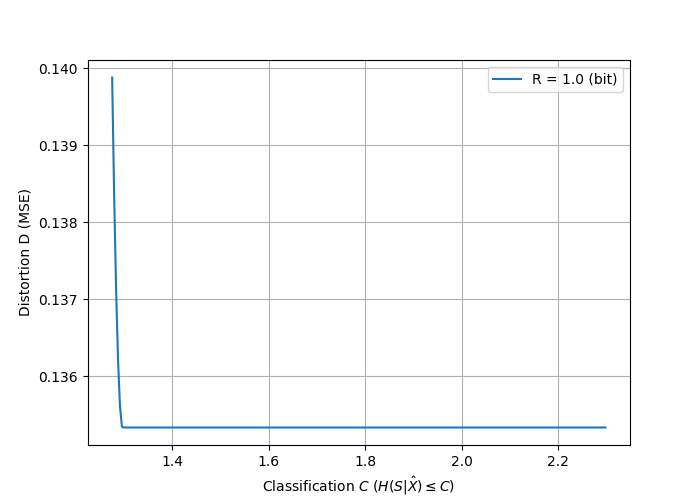}
        \label{fig:DCR}
    } 
    \subfigure[$D(C,R)$ functions across multiple rates.]{% 
        \includegraphics[width=0.48\textwidth]{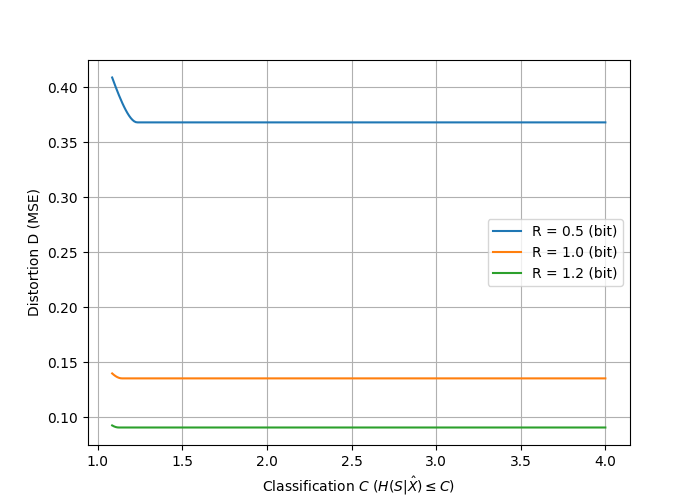}
        \label{fig:DCR_Mutltiple_Rate}
    } 
    \caption{Illustration of the distortion-classification-rate functions: 
    \subref{fig:DCR} shows the DCR function at a fixed rate, and \subref{fig:DCR_Mutltiple_Rate} shows how the function changes across multiple rates.}
    \label{fig:CDR_Function}
\end{figure}

Figure \ref{fig:CDR_Function} shows the information distortion-classification-rate function in Theorem \ref{TheoremDCR_GS}, with parameters set to \( \sigma_X^2 = \sigma_S^2 = 1.0 \) and \( \theta_1 = 0.7 \). Subfigure \ref{fig:CDR_Function}\subref{fig:DCR} depicts the DCR trade-off at a fixed rate, while subfigure \ref{fig:CDR_Function}\subref{fig:DCR_Mutltiple_Rate} illustrates how the function varies across multiple rate values.

These plots illustrate the operational interpretation of the DCR surface. At a fixed rate, relaxing the classification constraint (i.e., increasing \( C \)) enables lower achievable distortion. Conversely, for a fixed classification constraint, increasing the rate allows the encoder to preserve more source information, thereby reducing distortion. This visualization highlights the fundamental trade-off among rate, distortion, and classification, and identifies regimes where further relaxing classification requirements yields diminishing returns in distortion.

%===========================================================================================
\section{Convexity and Bounds of the Distortion-Classification-Rate Function for a General Source}
\label{sec:ConvexityBounds}
In this section, we investigate the fundamental properties of the distortion-classification-rate function for a general source. We begin by establishing its convexity and monotonicity, which are essential for understanding optimization behavior under joint distortion and classification constraints. We then derive a sharp upper bound on the distortion increase required to achieve optimal classification performance, providing insight into the inherent trade-offs between these objectives.

In \cite{Wang2024}, the authors showed that the rate-distortion-classification function \( R(D, C) \) is convex in \( D \) and \( C \) for a general source. Similarly, we establish the convexity of the distortion-classification-rate function \( D(C, R) \) in the following theorem.
\begin{theorem}\label{Theorem_DCR_Convexity}
The function \( D(C, R) \), defined for all points \( (C, R) \) where \( D(C, R) < +\infty \), satisfies the following properties:
\begin{itemize}
    \item It is monotonically non-increasing in both \( C \) and \( R \);
    \item It is convex.
\end{itemize}
\end{theorem}
\begin{proof}
A proof based on \cite{Wang2024} is provided in Appendix \ref{Appendix_Proof_DCR_Convexity}.
\end{proof}

Building on the proof of Theorem~\ref{Theorem_DCR_Convexity}, it can be similarly shown that the function \( C(D, R) \) is convex and monotonically non-increasing over the domain where \( C(D, R) < +\infty \). This duality reveals a structural symmetry between distortion and classification objectives under a fixed rate constraint. In the following, we establish bounds on the classification-distortion function.
\begin{figure}[!htbp]
\centering
\includegraphics[width=0.65\textwidth]{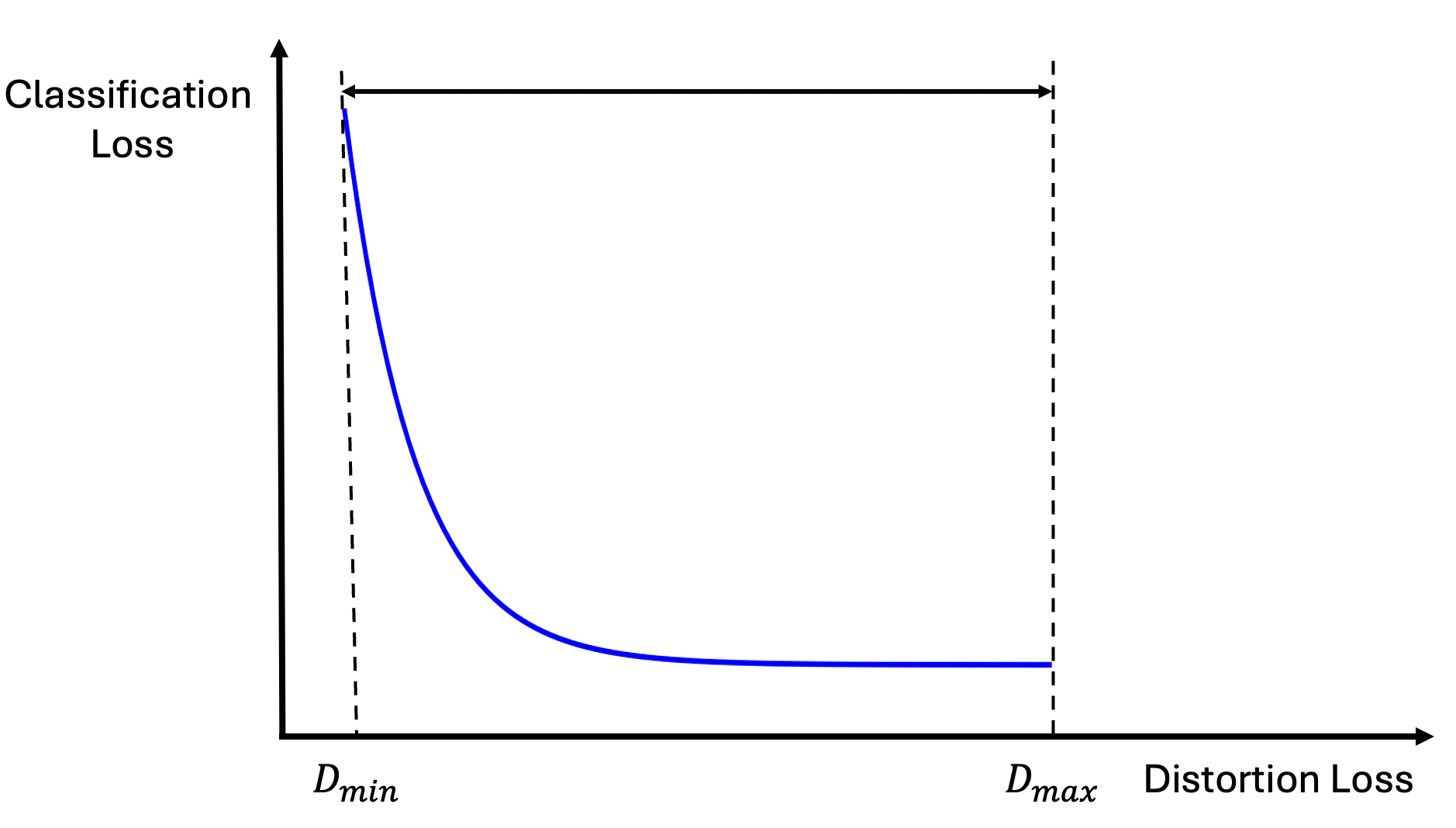}
\caption{Bounding the classification-distortion-rate function: the gap between \( D_{\min} \) and \( D_{\max} \) quantifies the additional distortion required to achieve the minimum classification loss. For the MSE distortion, Theorem \ref{Theorem_CDR_Bound} establishes that this increase is bounded by a factor of 2, corresponding to a 3~dB drop in PSNR.}
\label{Fig:Bound_CDR}
\end{figure}
First, for a given rate \( R(D, C) \), we seek to answer the following fundamental question: \emph{What is the minimal attainable distortion \( D_{\min} \) achievable by an optimal lossy compression scheme without considering the classification constraint?} This corresponds to the value:
\begin{mini!}|s|[2]
    {p_{\hat{X}|X}}
    {\mathbb{E}[\Delta(X, \hat{X})]}
    {\label{D_min}}
    {D_{\min} = }
\end{mini!}
which represents the horizontal coordinate of the leftmost point on the classification-distortion function.

Second, a complementary and critical question arises: \emph{What minimal distortion can be achieved if we strictly enforce the optimal classification constraint?} Formally, this corresponds to the optimization:
\begin{mini!}|s|[2]
    {p_{\hat{X}|X}}
    {\mathbb{E}[\Delta(X, \hat{X})]}
    {\label{D_max}}
    {D_{\max} = }
    \addConstraint{H(S|\hat{X}) = C_{\min}}
\end{mini!}
which identifies the horizontal coordinate of the rightmost point on the classification-distortion function, with \( C_{\min} \) representing the minimal achievable classification loss \( H(S | \hat{X}) \) (illustrated in Figure \ref{Fig:Bound_CDR}).

The quantities \( D_{\min} \) and \( D_{\max} \) characterize the extremes of achievable distortion at a fixed encoding rate. Specifically, \( D_{\min} \) corresponds to purely fidelity-driven compression, while \( D_{\max} \) reflects the distortion incurred when prioritizing classification accuracy. Understanding the gap between these extremes is essential for assessing the practical impact of classification constraints in lossy compression systems.
\begin{theorem}[Bound on the Classification-Distortion-Rate Function]\label{Theorem_CDR_Bound}
For a given rate \( R(D, C) \) and mean squared error distortion, the following inequality holds:
\begin{equation}
D_{\max} \leq 2 D_{\min},
\end{equation}
where \( D_{\min} \) and \( D_{\max} \) are defined in (\ref{D_min}) and (\ref{D_max}), respectively.

This bound is attained by an estimator \( \hat{X} \), characterized by the conditional distribution \( p_{\hat{X}|X} \), achieving the minimum classification loss \( H(S | \hat{X}) = C_{\min} \) while incurring a maximum mean squared error distortion of exactly \( 2 D_{\min} \).
\end{theorem}
\begin{proof}
The detailed proof is provided in Appendix \ref{Appendix_Proof_CDR_Bound}.
\end{proof}

This key theoretical result shows that, under MSE distortion, enforcing optimal classification increases distortion by at most a factor of two. In practical terms, this translates to a worst-case degradation of 3 dB in peak signal-to-noise ratio (PSNR). As a result, practitioners can balance classification performance and reconstruction fidelity with confidence that the loss in perceptual quality remains bounded.

%===========================================================================================
\section{Universal Representations}
In many practical scenarios, designing separate encoding schemes tailored to each specific combination of distortion and classification constraints is inefficient and costly. This motivates the exploration of \emph{universal representations}, where a single encoder is utilized to serve multiple distortion-classification constraints simultaneously. This section formally defines the universal representation framework, characterizes the rate penalty associated with universality, and provides detailed theoretical results for Gaussian and general sources.

\subsection{Definitions}
When the rate-distortion-classification function is interpreted as the minimum rate required to satisfy a given distortion-classification constraint pair \( (D, C) \) through joint optimization of both the encoder and the decoder, \emph{universal} RDC framework extends this setting by fixing the encoder and allowing only the decoder to adapt. This enables the system to accommodate multiple pairs of constraint \( (D, C) \in \Theta \), where \( \Theta \) is a specified set of such pairs, as illustrated in Figure~\ref{fig:Universal}.
\begin{figure}[h]
\centering
\includegraphics[width=0.6\textwidth]{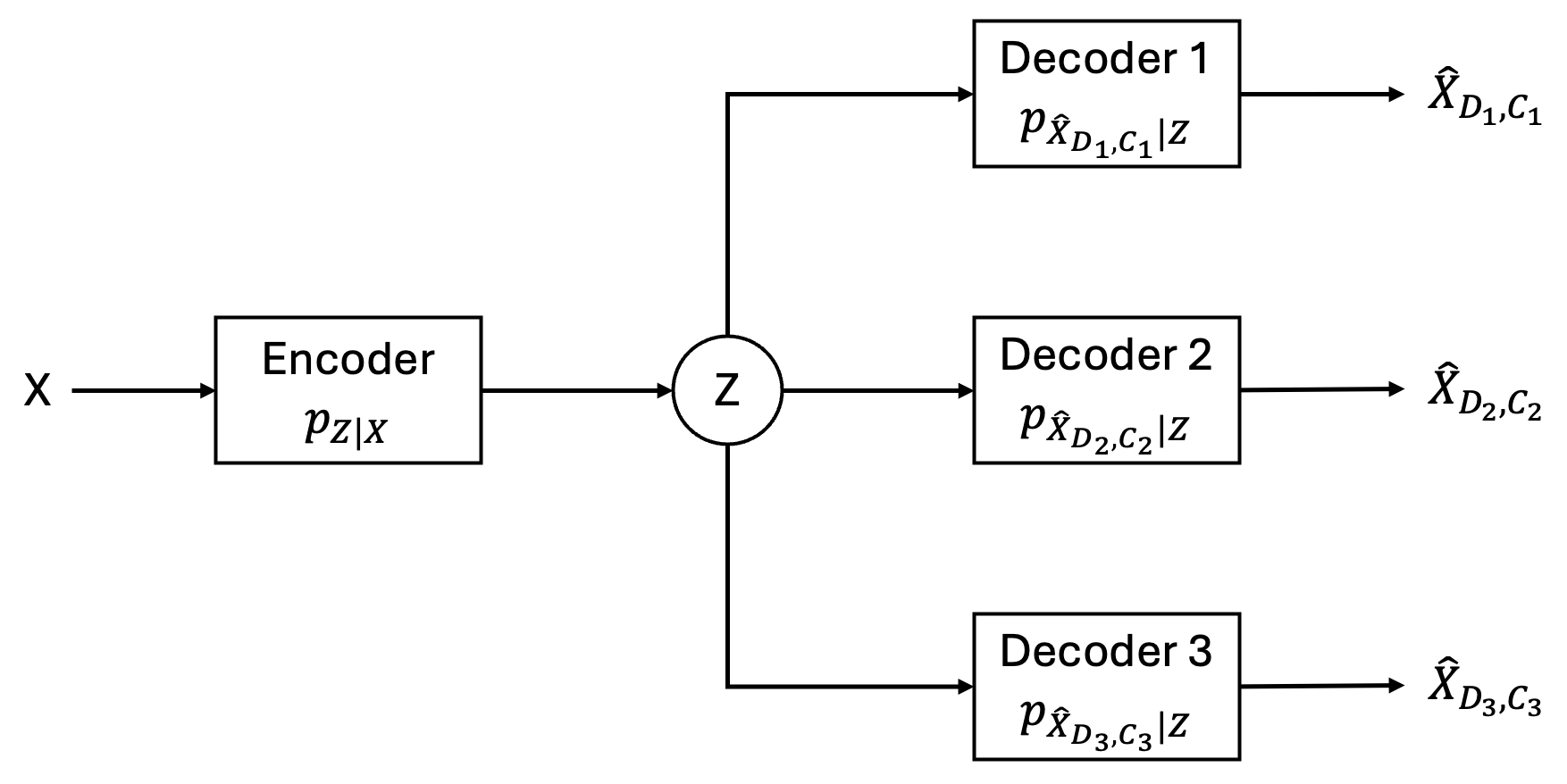}
\caption{Illustration of the universal representation framework.}
\label{fig:Universal}
\end{figure}

A fascinating scenario arises when \( \Theta \) contains all \( (D, C) \) pairs that lie along the RDC curve at a fixed rate. This leads to a fundamental question: \emph{How much additional rate is required to simultaneously satisfy all constraints in \( \Theta \) using a single encoder, rather than customizing the encoder for each target pair?} Ideally, the cost of such universality is minimal, meaning that the additional rate required to support all tasks is close to that needed for the most demanding one.

To formalize this notion, we introduce \emph{information universal rate-distortion-classification function}. We assume that \( X \) is a random variable and \( \Theta \) is a non-empty set of constraint pairs \( (D, C) \).

\begin{definition}[Information Universal Rate-Distortion-Classification Function]
Let \( Z \) be a representation of \( X \), generated through a conditional distribution \( p_{Z|X} \). Define \( \mathcal{P}_{Z|X}(\Theta) \) as the set of such transformations for which, for every \( (D, C) \in \Theta \), there exists a decoder \( p_{\hat{X}_{D,C}|Z} \) such that $\mathbb{E}[\Delta(X, \hat{X}_{D,C})] \leq D \quad \text{and} \quad H(S | \hat{X}_{D,C}) \leq C$, with the Markov chain \( X \leftrightarrow Z \leftrightarrow \hat{X}_{D,C} \) holding. The \emph{information universal RDC function} is defined as
\begin{equation}\label{eqn:R_phi}
R(\Theta) = \inf_{p_{Z|X} \in \mathcal{P}_{Z|X}(\Theta)} I(X; Z).
\end{equation}
\end{definition}

We refer to \( R(\Theta) \) as the information universal RDC function. A representation \( Z \) is said to be \(\Theta\)-universal with respect to \( X \) if it allows all constraints in \( \Theta \) to be satisfied using appropriate decoders. In this context, each decoder \( p_{\hat{X}_{D,C}|Z} \) maps the shared representation \( Z \) to a reconstruction \( \hat{X}_{D,C} \) tailored to the specific constraint pair.

The quantity \( R(\Theta) \) thus captures the minimal rate needed to meet \emph{all} constraints in \( \Theta \) with a shared encoder. In contrast, the term $
\sup_{(D, C) \in \Theta} R(D, C)$
represents the rate required to satisfy only the most demanding individual constraint. The difference between these two defines the \emph{rate penalty}.
\begin{definition}[Rate Penalty]
The \emph{rate penalty} for a constraint set \( \Theta \) is defined as
\begin{equation}
    A(\Theta) = R(\Theta) - \sup_{(D, C) \in \Theta} R(D, C),
\end{equation}
which quantifies the additional rate incurred when using a single encoder to satisfy all constraints in \( \Theta \).
\end{definition}

Let \( \Omega(R) = \{ (D, C) : R(D, C) \leq R \} \) denote the set of distortion-classification pairs that are achievable at rate \( R \). Ideally, we would like \( A(\Omega(R)) = 0 \) for all \( R \), indicating that the entire trade-off curve can be achieved without incurring any additional cost for universality. This would eliminate the need to design separate encoders for different distortion-classification goals at the same rate.

Given a representation \( Z \), we define the associated \emph{achievable distortion-classification region} as
\[
\Omega(p_{Z|X}) = \left\{ (D, C) : \exists\, p_{\hat{X}_{D,C}|Z} \text{ such that } \mathbb{E}[\Delta(X, \hat{X}_{D,C})] \leq D,\quad H(S | \hat{X}_{D,C}) \leq C \right\}.
\]
Intuitively, \( \Omega(p_{Z|X}) \) describes the set of constraint pairs that can be satisfied using the fixed representation \( Z \). If a representation \( Z \) satisfies \( I(X; Z) = R \) and \( \Omega(p_{Z|X}) = \Omega(R) \), then \( Z \) is said to achieve the \emph{maximal distortion-classification region} at rate \( R \). That is, for any other representation \( Z' \) with \( I(X; Z') \leq R \), we have $
\Omega(p_{Z'|X}) \subseteq \Omega(p_{Z|X})$.

%-------------------------------------------------------------------------------------------
\subsection{Main Results}
\begin{theorem}[No Rate Penalty for a Gaussian Source]\label{Theorem_gaussian_universality}
Let \( X \sim \mathcal{N}(\mu_X, \sigma_X^2) \) be a scalar Gaussian source and let \( S \sim \mathcal{N}(\mu_S, \sigma_S^2) \) be a classification variable with covariance \( \mathrm{Cov}(X, S) = \theta_1 \). Assume that the distortion is measured using the mean squared error and that the classification loss is measured via conditional entropy. Let \( \Theta \) denote any non-empty set of distortion-classification constraint pairs \( (D, C) \). Then,
\begin{equation}
    A(\Theta) = 0,
\end{equation}
which implies that satisfying the most demanding constraint in \( \Theta \) is sufficient to simultaneously satisfy all others using a fixed encoder and there is no rate penalty for universality in this case.

Furthermore, consider any representation \( Z \) that is jointly Gaussian with \( X \) and satisfies
\begin{equation}
    I(X; Z) = \sup_{(D, C) \in \Theta} R(D, C).
\end{equation}
Then the following inclusion holds:
\begin{equation}\label{eqn:sup_rate}
    \Theta \subseteq \Omega(p_{Z|X}) = \Omega(I(X; Z)),
\end{equation}
meaning that \( Z \) achieves the maximal distortion-classification region at rate \( I(X; Z) \); i.e., all constraints in \( \Theta \) are simultaneously achievable via appropriate decoders applied to a common representation \( Z \).
\end{theorem}
\begin{proof} 
The proof is provided in Appendix \ref{Appendix_Proof_GS_Universality}.
\end{proof}

\begin{proposition}[Equivalence Between Zero Rate Penalty and Full Distortion-Classification Region]
\label{Prop_equivalence}
Assume the following regularity conditions are satisfied:
\begin{itemize}
    \item \( \sup_{(D, C) \in \Omega(R)} R(D, C) = R' \);
    \item The infimum in the definition of \( R(\Omega(R')) \) is achieved.
\end{itemize}
Then, the rate penalty satisfies \( A(\Omega(R')) = 0 \) if and only if there exists a representation \( Z \) with \( I(X; Z) = R' \) such that $\Omega(p_{Z|X}) = \Omega(I(X; Z))$.
\end{proposition}
\begin{proof}
The proof is provided in Appendix~\ref{Appendix_Proof_prop_equivalence}.
\end{proof}

In addition, we consider a general source \( X \sim p_X \) and characterize the distortion-classification region induced by an arbitrary representation \( Z \) under MSE distortion.

\begin{theorem}[Achievable Universality Region for a General Source]\label{Theorem_general_universality}
Consider a general source \( X \sim p_X \) and a classification variable \( S \), with covariance \( \mathrm{Cov}(X, S) = \theta_1 \). Assume that distortion is measured using mean squared error, and classification loss is measured using conditional entropy, $H(S | \hat{X})$. Let \( Z \) be an arbitrary representation of \( X \), and define the minimum mean square error (MMSE) estimator of \( X \) given \( Z \) as \( \tilde{X} = \mathbb{E}[X | Z] \). Then, the closure of this region, denoted \( \mathrm{cl}\big(\Omega(p_{Z \mid X})\big) \), satisfies
\begin{equation*}
\Omega(p_{Z|X}) \subseteq 
\left\{ (D, C) : 
    D \geq \mathbb{E}[\|X - \tilde{X}\|^2] + 
    \begin{aligned}
        &\inf_{p_{\hat{X}}} \quad W^2_2(p_{\tilde{X}}, p_{\hat{X}}) \\
        &\text{s.t} \quad H(S | \hat{X}) \leq C
    \end{aligned}
\right\}
\subseteq \mathrm{cl}(\Omega(p_{Z|X})),
\end{equation*}
where \( W^2_2(\cdot, \cdot) \) denotes the squared 2-Wasserstein distance is defined as  
\begin{equation*}
\begin{aligned}
 W_{2}^2(p_X, p_{\hat{X}}) = \underset{p_{X,\hat{X}}} \inf \,\, &\mathbb{E}[\|X - \hat{X}\|^2] \\
 \text{s.t.} \,\,\,\,\,\,\,\,\,\, & \int_{-\infty }^{\infty} p_{X,\hat{X}}dX = p_{\hat{X}}, \int_{-\infty }^{\infty} p_{X,\hat{X}}d\hat{X} = p_{X}.
\end{aligned}
\end{equation*} 

Furthermore, the closure \( \mathrm{cl}(\Omega(p_{Z|X})) \) contains the following two extreme points:
\begin{align*}
(D^{(a)}, C^{(a)}) &= \left( \mathbb{E}[\|X-\tilde{X}\|^2], \sum_{s}\sum_{\tilde{x}} p_{\tilde{X}} p_{S|\tilde{X}} \log\frac{1}{p_{S|\tilde{X}}} \right), \\
(D^{(b)}, C^{(b)}) &= \left( \mathbb{E}[\|X - \tilde{X}\|^2] + W_2^2(p_{\tilde{X}}, p_{\hat{X}}^{(U)}), C_{\min}^{(\text{U})} \right),
\end{align*}
where \( C_{\min}^{(\text{U})} \) is the minimum classification loss achievable at a given distortion level under universality, defined by the following optimization problem. With the given rate, $R(\Theta)$, we have:
\begin{mini!}|s|[2]
    {p_{\hat{X}}} 
    {H(S | \hat{X})} 
    {\label{C_min}} 
    {p_{\hat{X}}^{C_{\min}^{(\text{U})}} = \arg} 
    \addConstraint{\mathbb{E}[\|X - \hat{X}\|^2]}{\leq D.}
\end{mini!}
The corresponding minimum classification loss is given by
\begin{equation*}
    C_{\min}^{(\text{U})} = \sum_{s}\sum_{\hat{x}^{(U)}} p_{\hat{X}}^{(U)} p_{S|\hat{X}^{(U)}} \log\frac{1}{p_{S|\hat{X}^{(U)}}}.
\end{equation*}
\end{theorem}
\begin{proof} 
The proof is provided in Appendix \ref{Appendix_Proof_General_Universality}.
\end{proof}
\begin{figure}[!htbp]
\centering
\includegraphics[width=0.7\textwidth]{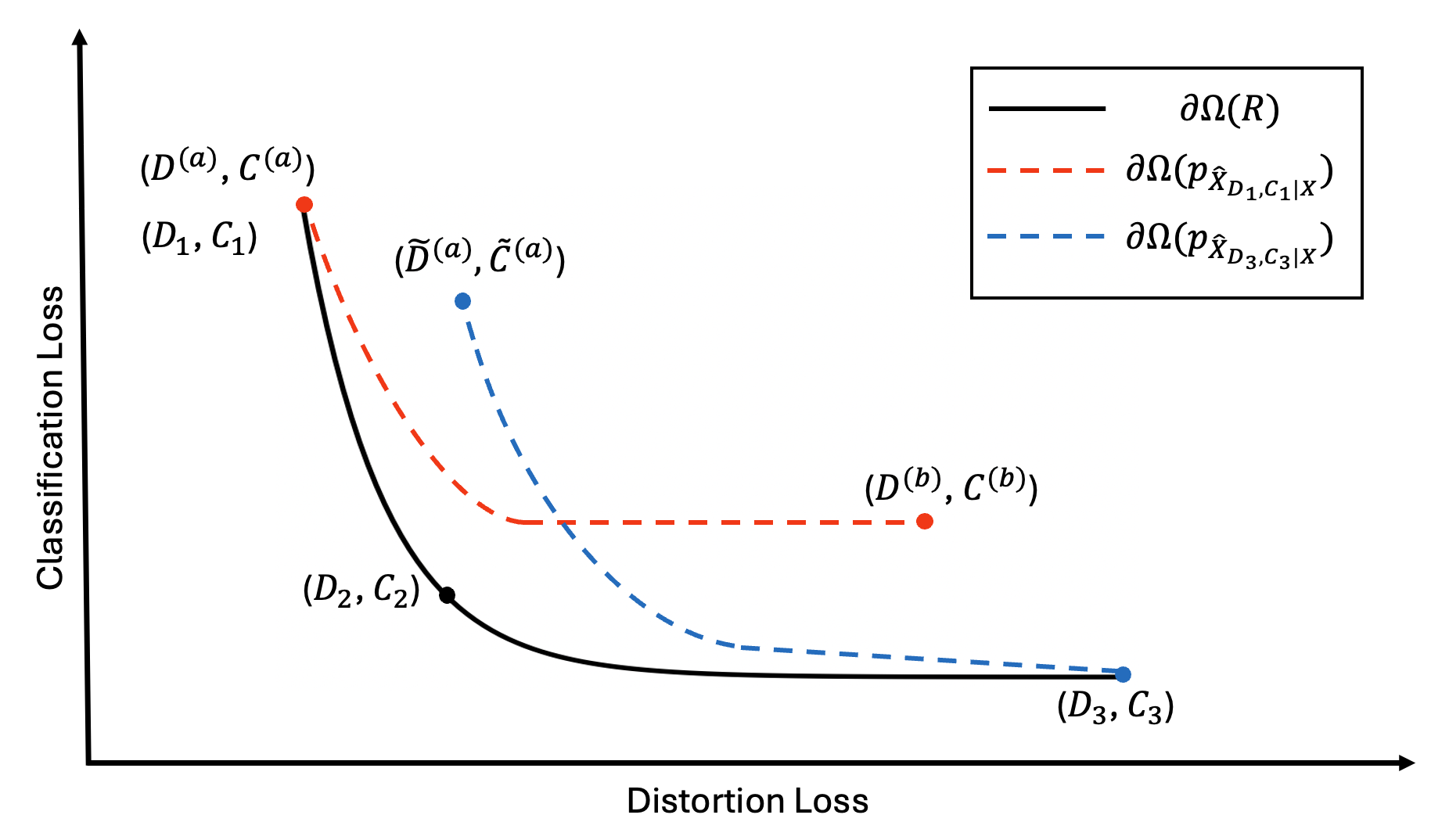}
\caption{Universality for a general source. Shown are the boundaries of achievable distortion-classification regions corresponding to three different representations: the minimal distortion point \( (D_1, C_1) \), where \( R(D_1, C_1) = R(D_1, \infty) \); the midpoint \( (D_2, C_2) \); and the minimal classification loss point \( (D_3, C_3) \), where \( C_3 = C_{\min}^{\text{(Conv)}} \). Also illustrated are the two extreme points \( (D^{(a)}, C^{(a)}) \) and \( (D^{(b)}, C^{(b)}) \) associated with the representation \( \hat{X}_{D_1,C_1} \), where \( (D^{(a)}, C^{(a)}) \) coincides with \( (D_1, C_1) \).}
\label{Fig:Universal_GernalSource_CDR}
\end{figure}

In addition, to gain further insight into the structure of achievable regions, let \( Z \) be the optimal reconstruction \( \hat{X}_{D,C} \) associated with a point \( (D, C) \) on the distortion-classification trade-off curve for a given rate \( R \); that is, $
I(X; \hat{X}_{D,C}) = R(D, C) = R$,
and the distortion and classification loss satisfy $
\mathbb{E}[\|X - \hat{X}_{D,C}\|^2] = D, \quad H(S | \hat{X}_{D,C}) = C$.
We assume, for simplicity, that such an optimal reconstruction \( \hat{X}_{D,C} \) exists for every \( (D, C) \) on the tradeoff boundary, and that decreasing either \( D \) or \( C \) would violate the constraint \( R(D, C) = R \). Under this assumption, the point \( (D, C) \) lies on the boundary of the closure \( \mathrm{cl}(\Omega(p_{\hat{X}_{D,C} \mid X})) \).

According to Theorem~\ref{Theorem_general_universality}, the set \( \mathrm{cl}(\Omega(p_{\hat{X}_{D,C} \mid X})) \) includes two extreme points: the upper-left corner \( (D^{(a)}, C^{(a)}) \), corresponding to minimal distortion, and the lower-right corner \( (D^{(b)}, C^{(b)}) \), corresponding to minimal classification loss. Moreover, this closure defines a convex region that includes all these points. Figure~\ref{Fig:Universal_GernalSource_CDR} illustrates both \( \Omega(R) \) and the achievable region \( \Omega(p_{\hat{X}_{D,C} \mid X}) \) for various representative points \( (D, C) \) on the trade-off curve. The following theorem provides a quantitative characterization of this structure.

\begin{theorem}[Quantitative Characterization for a General Source]
\label{Theorem_Quantitative_Results}
Let \( \hat{X}_{D_1, C_1} \) be the optimal reconstruction associated with the point \( (D_1, C_1) \) on the rate-distortion-classification trade-off curve, i.e., in the conventional sense where $
I(X; \hat{X}_{D_1, C_1}) = R(D_1, C_1)$.
Then, the upper-left extreme point of the achievable region \( \Omega(p_{\hat{X}_{D_1, C_1} \mid X}) \) coincides with that of \( \Omega(R) \); specifically, $
(D^{(a)}, C^{(a)}) = (D_1, C_1)$. Consider the lower-right extreme points \( (D^{(b)}, C^{(b)}) \in \Omega(p_{\hat{X}_{D_1, C_1} \mid X}) \) and \( (D_3, C_3) \in \Omega(R) \), where \( C_3 = C_{\min}^{\text{(Conv)}} \) and \( R(D_3, C_3) = R(D_1, \infty) \). The distortion gap between these two points satisfies the following inequality:
\begin{equation}
D_3 - D^{(b)} \geq \sigma_X^2 + \sigma_{\hat{X}_{D_3, C_3}}^2 
- 2 \sigma_{\hat{X}_{D_3, C_3}} \sqrt{\sigma_X^2 - D_1} - 2D_1,
\end{equation}
and the corresponding distortion ratio is lower bounded as:
\begin{equation}
\frac{D_3}{D^{(b)}} \geq 
\frac{\sigma_X^2 + \sigma_{\hat{X}_{D_3, C_3}}^2 
- 2 \sigma_{\hat{X}_{D_3, C_3}} \sqrt{\sigma_X^2 - D_1}}{2D_1}.
\end{equation}

Moreover, in the special case where the squared 2-Wasserstein distance between \( p_X \) and \( p_{\hat{X}_{D_3, C_3}} \) is zero, i.e., \( W_2^2(p_X, p_{\hat{X}_{D_3, C_3}}) = 0 \), which implies \( \sigma_X^2 = \sigma_{\hat{X}_{D_3, C_3}}^2 \), the distortion gap becomes negligible under the following conditions:
\begin{equation}
D_3 - D^{(b)} \approx 0 \quad \text{when } D_1 \approx 0 \text{ or } D_1 \approx \sigma_X^2,
\end{equation}
\begin{equation}
\frac{D_3}{D^{(b)}} \approx 1 \quad \text{when } D_1 \approx \sigma_X^2.
\end{equation}
\end{theorem}
\begin{proof} 
The proof is provided in Appendix \ref{Appendix_Proof_Quantitative_Results}.
\end{proof}

We previously focused on the one-shot setting. We now extend our analysis to the asymptotic regime, where an i.i.d. sequence \( X^n \) is jointly encoded, with each symbol drawn from the marginal distribution \( p_X \).

\begin{definition}[Asymptotic Universal Rate-Distortion-Classification Function]
Let \( \Theta \) be a non-empty set of distortion-classification constraint pairs \( (D, C) \). A representation is said to be \emph{\( \Theta \)-universal} with asymptotic rate \( R \) if there exists a sequence of random variables \( U^{(n)} \), encoding functions \( f^{(n)}_{U^{(n)}} : \mathcal{X}^n \rightarrow \mathcal{C}^{(n)}_U \), and decoding functions \( g^{(n)}_{U^{(n)}, D, C} : \mathcal{C}^{(n)}_U \rightarrow \hat{\mathcal{X}}^n \) for each \( (D, C) \in \Theta \), such that:
\begin{align}
    \frac{1}{n} \sum_{i=1}^n \mathbb{E}[\Delta(X(i), \hat{X}_{D,C}(i))] &\leq D, \label{eq:constraintD}\\
    H\left( S \Big| \frac{1}{n} \sum_{i=1}^n \hat{X}_{D,C}(i) \right) &\leq C, \label{eq:constraintP}
\end{align}
and the expected code length satisfies:
\begin{equation*}
    \limsup_{n \to \infty} \frac{1}{n} \mathbb{E}[\ell(f^{(n)}_{U^{(n)}}(X^n))] \leq R,
\end{equation*}
where \( \hat{X}^n_{D,C} \triangleq g^{(n)}_{U^{(n)}, D, C}(f^{(n)}_{U^{(n)}}(X^n)) \). The infimum of such rates \( R \) over all valid constructions is denoted as \( R^{(\infty)}(\Theta) \).
\end{definition}

We now state the following result, which establishes that the asymptotic formulation is consistent with the information-theoretic definition of universality.

\begin{theorem}[Asymptotic Universality] \label{theorem_asymptotic_universality}
For any set \( \Theta \), the asymptotic universal rate-distortion-classification function satisfies
\begin{equation*}
    R^{(\infty)}(\Theta) = R(\Theta).
\end{equation*}
\end{theorem}
\begin{proof}
The proof is provided in Appendix~\ref{Appendix_Proof_asymptotic_universality}.
\end{proof}

%===========================================================================================
\section{Experimental Results}
\subsection{Illustration of No Rate Penalty for a Gaussian Source}
\begin{figure*}[h] 
    \centering
    \includegraphics[width=0.65\textwidth]{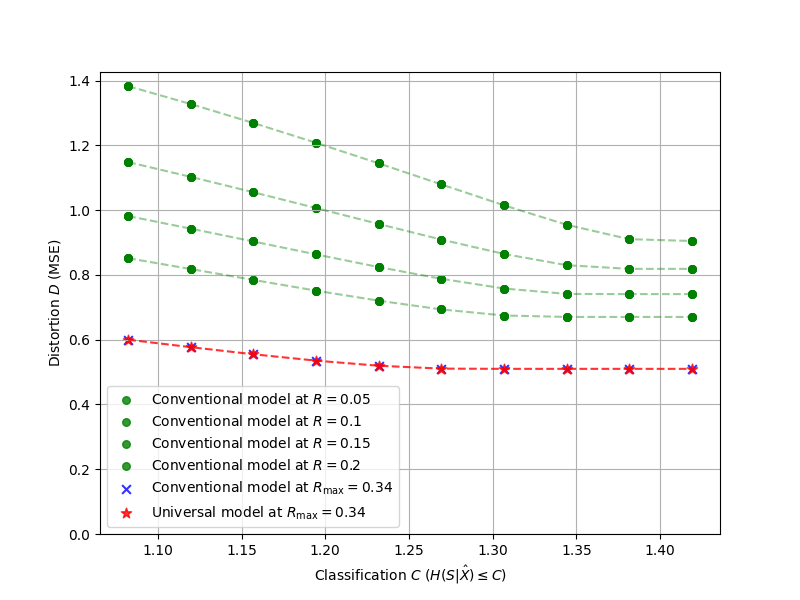}
    \caption{Classification-distortion-rate function at various rates for a Gaussian source. The universal model achieves the same region as the conventional model, verifying Theorem~\ref{Theorem_gaussian_universality}.}
    \label{fig:DCR_Gaussian_NoPenalty}
\end{figure*}
We numerically validate the absence of a rate penalty when transitioning from multiple rate-specific Gaussian encoders (conventional model) to a single universal encoder. The experiment considers a scalar Gaussian source \(X \sim \mathcal{N}(0, \sigma_X^2)\) and a classification variable \(S \sim \mathcal{N}(0, \sigma_S^2)\), correlated via coefficient \(\rho\). The goal is to evaluate the classification-distortion-rate function under mean squared error distortion and classification constraint \(C = H(S \mid \hat{X})\).

We fix the variances \(\sigma_X^2 = \sigma_S^2 = 1.0\) and set \(\rho = 0.7\), resulting in a maximum achievable rate of \( R_{\max} = \frac{1}{2} \log \left( \frac{1}{1 - \rho^2} \right) = 0.34 \). The conventional model is evaluated at five rate values: \( [0.05, 0.1, 0.15, 0.2, 0.34] \), with the corresponding \((C, D)\) tradeoffs derived using Theorem~\ref{TheoremDCR_GS}. For each rate, we compute the achievable region by varying the decoder and measuring the resulting classification and distortion performance. The universal model, guided by Theorem~\ref{Theorem_gaussian_universality}, constructs a single high-rate Gaussian representation at \(R_{\max}\) and varies the decoder to explore the entire \((C,D)\) space.

Figure~\ref{fig:DCR_Gaussian_NoPenalty} displays the achievable \((C,D)\) regions under three settings. The green points correspond to the union of all conventional models, each operating at a different rate. The blue points denote a conventional encoder at \(R = R_{\max}\). The red points represent the universal encoder, also fixed at \(R_{\max}\), with varying decoders. Notably, the red points align precisely with the blue boundary, demonstrating that the universal model achieves the full CDR region without additional rate overhead. This confirms Theorem~\ref{Theorem_gaussian_universality} and affirms the feasibility of universal representations for Gaussian sources under classification constraints.

%-------------------------------------------------------------------------------------------
\subsection{Universal Representation for Lossy Compression}
The rate-distortion-classification tradeoff was observed in deep learning-based image compression when classifier regularization was integrated into the training pipeline~\cite{Wang2024, Zhang2023}. In such settings, achieving a specific point in the tradeoff space typically requires training a dedicated end-to-end model that jointly optimizes the encoder and decoder for that objective. However, this approach is often computationally expensive and inflexible for practical deployment.

To address this limitation, an alternative strategy is to reuse a pre-trained encoder while adapting only the decoder or classifier to meet varying task requirements. We refer to models that jointly train both encoder and decoder for each objective as \textit{conventional models}, and those that retain a fixed encoder and retrain only the decoder as \textit{universal models}. In our framework, universal models leverage encoders originally trained under the conventional paradigm. Under identical datasets and hyperparameters, the only difference between the two lies in whether the encoder parameters are updated during training.
\begin{figure}[!htbp]
\center \includegraphics[width=0.95\textwidth]{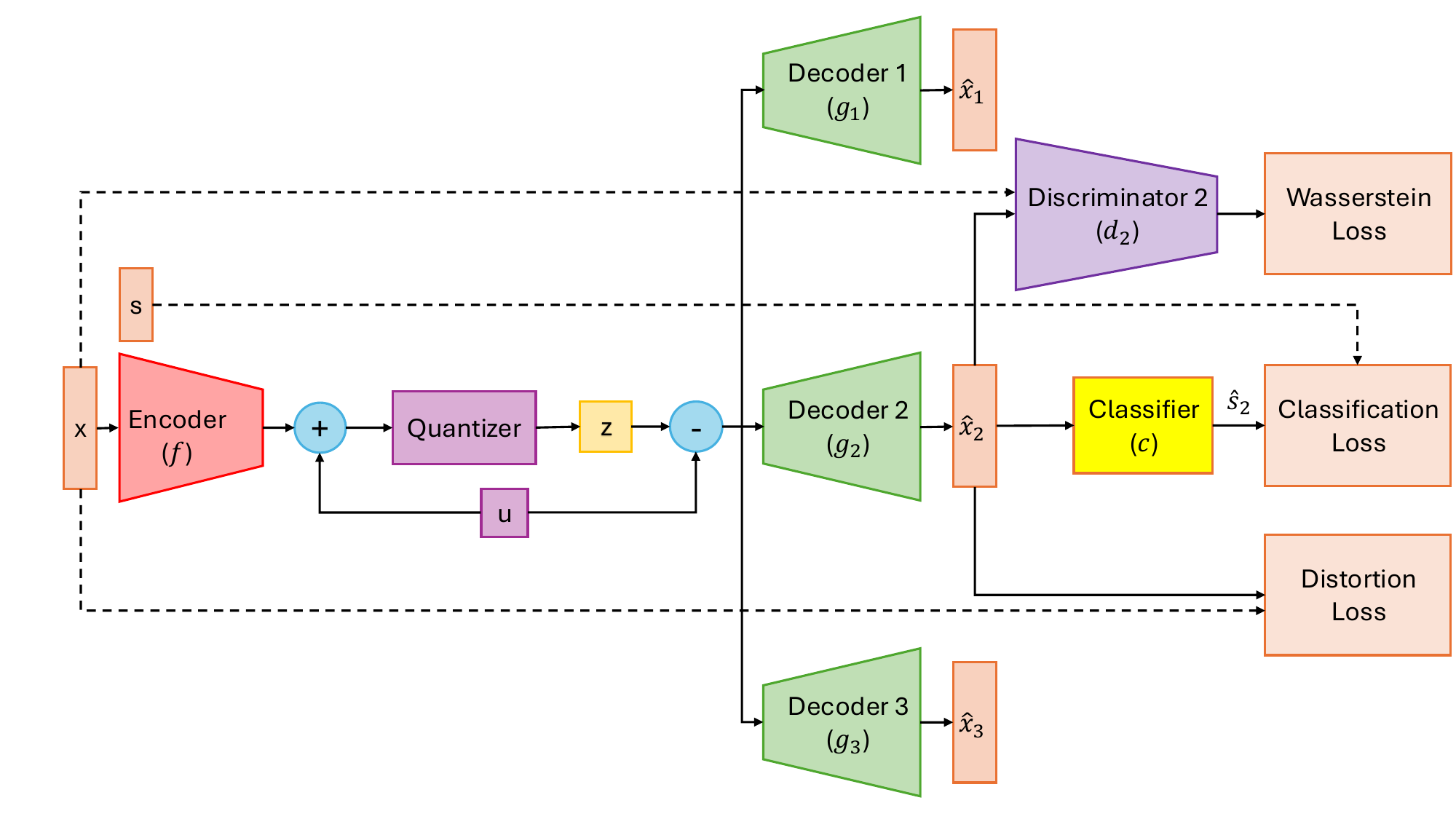}
\caption{Figure illustrates the universal model setup. A single encoder \( f \), trained for a specific classification-distortion tradeoff, is frozen and reused. Multiple decoders \( \{g_i\} \) and discriminators $\{d_i\}$ are then trained independently using the fixed latent representation \( z \). A shared randomness source \( u \) enables universal quantization, and a pre-trained classifier \( C \) evaluates classification performance.}
\label{fig:Scheme}
\end{figure}

%-------------------------------------------------------------------------------------------
\subsubsection{Training}
We use a stochastic autoencoder with a pre-trained classifier and GAN-based discriminator, consisting of an encoder \( f \), decoder \( g \), classifier \( c \), and discriminator \( d \). In the conventional setup, \( f \), \( g \), and \( d \) are trainable. The encoder maps input \( X \) to a latent representation \( f(X) \in [-1,1]^{\text{dim}} \) via a final \texttt{tanh} activation. This output is uniformly quantized into \( L \) levels per dimension, yielding an upper-bound compression rate of \( R = \text{dim} \times \log_2 L \), as established in~\cite{agustsson2019generative}. To perform quantization, we use dithered quantization~\cite{gray1993dithered, ziv1985universal}, assuming shared randomness \( U \sim \mathcal{U}[-1/(L{-}1), 1/(L{-}1)]^{\text{dim}} \). The encoder outputs: $ Z = \mathrm{Quantize}(f(X) + U)$ and the decoder reconstructs \( \hat{X} = g(Z - U) \). This approach centers quantization noise and enables gradient flow using the soft estimator from~\cite{mentzer2018conditional}.

The distortion loss is measured by MSE. The output \( \hat{X} \) is passed through the classifier \( c \) to produce the predicted label distribution \( \hat{S} = c(\hat{X}) \), with classification loss computed via cross-entropy \( \text{CE}(S, \hat{S}) \), an upper bound on conditional entropy, i.e. $H(S | \hat{X}) \leq \mathrm{CE}(S, \hat{S})$ ~\cite{boudiaf2021unifying_cross_entropy}. To estimate the squared 2-Wasserstein distance, both \( X \) and \( \hat{X} \) are input to \( d \), with Wasserstein-1 loss following the GAN strategy in~\cite{UniversalRDPs}.

The compression rate is upper bounded by \( \text{dim} \times \log_2(L) \), where \( \text{dim} \) is the encoder output size and \( L \) the quantization level. The total loss is:
\begin{equation}
\label{eqn:experimental_loss}
\mathcal{L} = \lambda_d \, \mathbb{E}[\|X - \hat{X}\|^2] + \lambda_c \, \text{CE}(S, \hat{S}) + \lambda_p \, W_1(p_X, p_{\hat{X}}),
\end{equation}
where \( \lambda_d \), \( \lambda_c \), and \( \lambda_p \) controlling the tradeoffs.

To construct the universal model, the trained encoder \( f \) is frozen, and a new decoder \( g_1 \) and discriminator \( d_1 \) are trained using:
\begin{equation}
\label{eqn:experimental_loss2}
\mathcal{L}_1 = \lambda_d^1 \, \mathbb{E}[\|X - \hat{X}_1\|^2] + \lambda_c^1 \, \text{CE}(S, \hat{S}) + \lambda_p^1 \, W_1(p_X, p_{\hat{X}}),
\end{equation}
where \( \lambda_d^1 \), \( \lambda_c^1 \), and \( \lambda_p^1 \) adjust the task-specific tradeoffs. A schematic of the full system is shown in Figure~\ref{fig:Scheme}.

%-------------------------------------------------------------------------------------------
\subsection{Results}
\begin{figure}[h]
    \centering
    \subfigure[The RDC curve under rate $R = 4.75$.]{
        \includegraphics[width=0.48\textwidth]{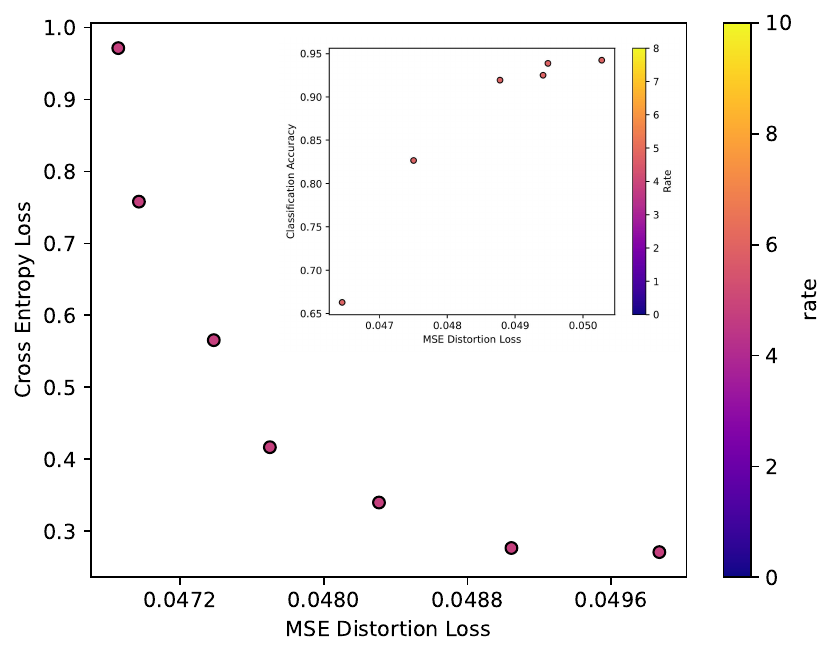}
        \label{fig:CDR_MNIST_Add_Accuracy}
    } 
    \hfill
    \subfigure[The RDC curves at multiple rates.]{
        \includegraphics[width=0.48\textwidth]{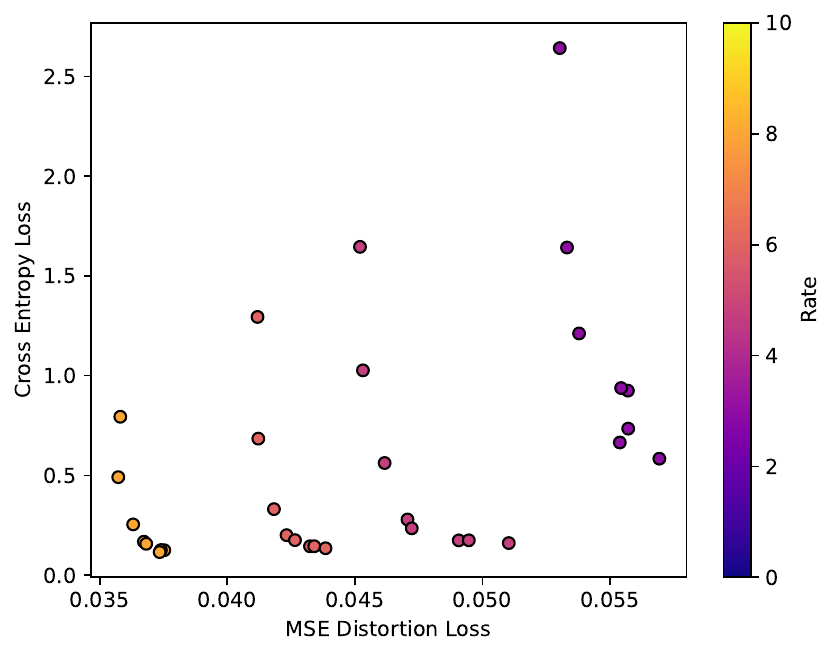}
        \label{fig:RDC_Mutiple_Rates_MNIST}
    } 
    \caption{The rate-distortion-classification tradeoffs on the MNIST dataset using conventional models. 
    \subref{fig:CDR_MNIST_Add_Accuracy} shows the tradeoff between distortion and classification, evaluated via either cross-entropy loss or classification accuracy. 
    \subref{fig:RDC_Mutiple_Rates_MNIST} illustrates the RDC curves across multiple encoding rates.}
    \label{fig:RDC_Conventional}
\end{figure}
Figure~\ref{fig:RDC_Conventional} illustrates the rate-distortion-classification functions of conventional models on the MNIST dataset. The tradeoff between distortion and classification performance, quantified via cross-entropy loss and classification accuracy, is clearly depicted in Figure~\ref{fig:CDR_MNIST_Add_Accuracy} at a fixed rate of \( R = 4.75 \), corresponding to an encoder output dimension of \(\text{dim} = 3\) and quantization level \( L = 3 \). Figure~\ref{fig:RDC_Mutiple_Rates_MNIST} presents RDC curves across multiple rate configurations. The evaluated (dim, \( L \)) combinations include (3,2), (3,3), (3,4), and (4,4), providing a range of bit budgets for comparison.

\begin{figure}[h]
    \centering
    \subfigure[RDC curve on the MNIST dataset.]{
        \includegraphics[width=0.48\textwidth]{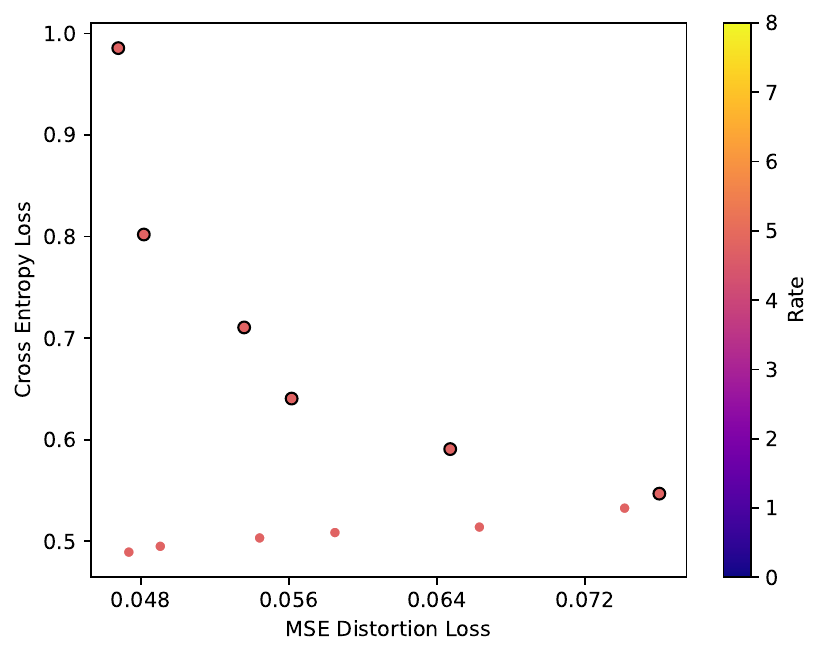}
        \label{fig:CDR_Comparision_Perception_MNIST}
    } 
    \hfill
    \subfigure[RDC curve on the SVHN dataset.]{
        \includegraphics[width=0.48\textwidth]{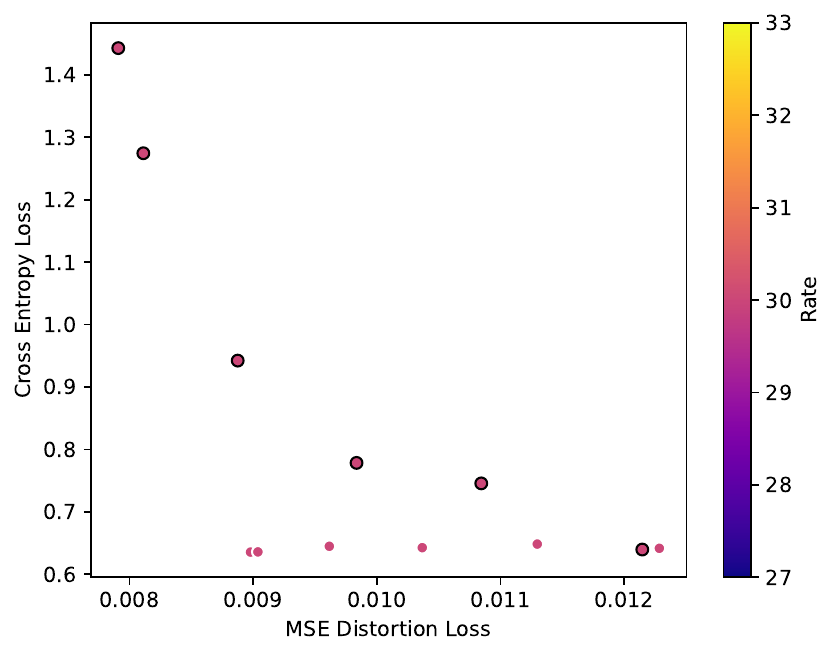}
        \label{fig:CDR_Comparision_Perception_SVHN}
    } 
    \caption{Rate-distortion-classification curves on the \subref{fig:CDR_Comparision_Perception_MNIST} MNIST ($R = 4.75$) dataset and \subref{fig:CDR_Comparision_Perception_SVHN} SVHN dataset ($R=30$). Points with black outlines correspond to conventional models jointly trained for specific classification-distortion objectives. The remaining points represent universal models, where decoders are trained on top of a frozen encoder optimized for low classification loss \( C \). 
    Despite the encoder being fixed, universal models achieve distortion levels comparable to conventional models, confirming that an encoder trained for low \( C \) can support diverse trade-offs via decoder retraining. These results support Theorem~\ref{Theorem_Quantitative_Results}.}
    \label{fig:RDC_No_Distortion}
\end{figure}
Each point in these figures corresponds to an encoder-decoder pair trained with a specific rate and loss configuration \( (\lambda_d, \lambda_c, \lambda_p) \). Points sharing the same color indicate models trained under the same rate. The results consistently reveal a tradeoff: achieving lower classification loss often comes at the cost of increased distortion. Moreover, as the rate \( R \) increases, the entire distortion-classification curve shifts downward and to the left, signifying that higher-rate encoders are able to achieve both lower distortion and better classification accuracy. This trend highlights the value of allocating more bits to support multi-task compression objectives.
\begin{figure}[h]
    \centering
    \subfigure[Decompressed outputs on the MNIST dataset.]{
        \includegraphics[width=0.48\textwidth]{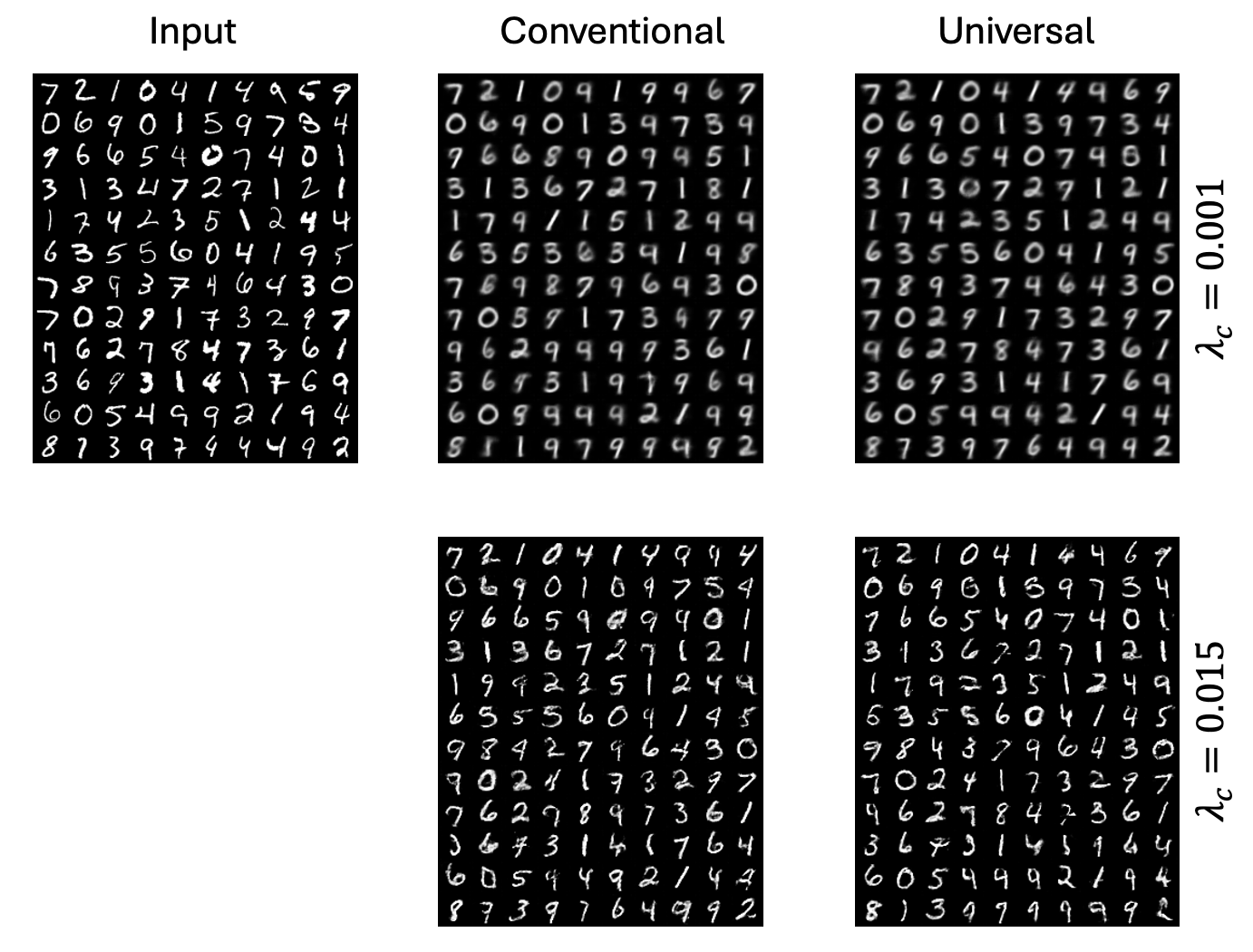}
        \label{fig:Decompression_MNIST}
    } 
    \hfill
    \subfigure[Decompressed outputs on the SVHN dataset.]{
        \includegraphics[width=0.48\textwidth]{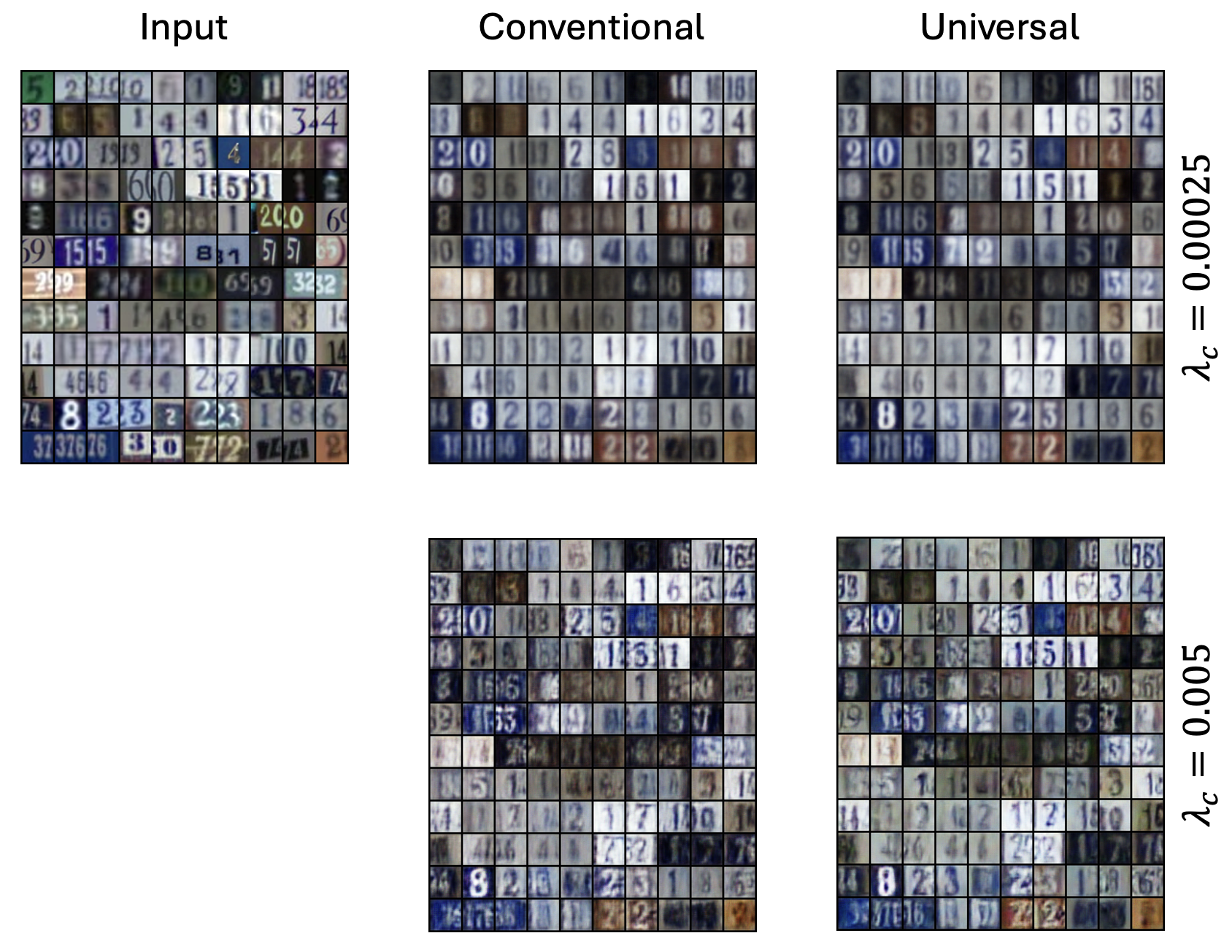}
        \label{fig:Decompression_SVHN}
    } 
    \caption{Visualizations of decompressed outputs from selected models on the MNIST dataset ($R = 4.75$) and SVHN dataset ($R = 30$). As the weight on the classification loss \( \lambda_c \) increases, the reconstructions become progressively blurrier, reflecting the tradeoff between classification accuracy and reconstruction fidelity.}
    \label{fig:Decompression}
\end{figure}

Figure~\ref{fig:RDC_No_Distortion} shows the rate-distortion-classification tradeoff on the  MNIST dataset ($R = 4.75$) and SVHN dataset ($R = 30$), obtained by varying loss coefficients. Black-outlined points represent the conventional model trained jointly for specific classification-distortion objectives. Other points correspond to the universal model with decoders trained on a fixed encoder optimized for low classification loss \( C \). As expected, reducing \( C \) increases distortion, reflecting a tradeoff between classification accuracy and reconstruction fidelity. Despite using a fixed encoder, the universal model achieves distortion levels comparable to the conventional model, confirming that an encoder trained for low \( C \) can still support diverse tradeoffs through decoder retraining. These observations support the validity of Theorem~\ref{Theorem_Quantitative_Results}.

However, a noticeable classification gap remains: universal decoders cannot recover low \( C \) performance if the encoder is trained only for high-distortion objectives. This highlights the decoder's limited generative capacity when the encoder fails to preserve the classification-task information. 

Figure~\ref{fig:Decompression} presents qualitative decompression results from selected models on the MNIST and SVHN datasets, using fixed compression rates of \( R = 4.75 \) and \( R = 30 \), respectively. These examples highlight the visual impact of increasing the classification loss weight \( \lambda_c \) during training. As \( \lambda_c \) increases, the encoder and decoder prioritize preserving semantic features relevant to classification, which leads to a noticeable degradation in image quality. In particular, the reconstructed digits and digits-over-backgrounds appear increasingly blurred, indicating a tradeoff between classification accuracy and perceptual fidelity. This visualization supports the notion that strong classification constraints may impair reconstruction quality, emphasizing the need for carefully balancing multiple objectives in learned lossy compression frameworks.

% % %===========================================================================================
% \section{Discussion}

%===========================================================================================
\section{Conclusion}
We proposed a universal rate-distortion-classification framework that enables a single encoder to support multiple task objectives through specialized decoders, removing the need for separate encoders per distortion-classification tradeoff. For the Gaussian source with MSE distortion, we proved that the full rate-distortion classification region is achievable with zero rate penalty using a fixed encoder. For the general source, we characterized the achievable region using MMSE estimation and the 2-Wasserstein distance, identifying conditions under which encoder reuse incurs negligible distortion penalty. Empirical results on the MNIST and SVHN datasets support our theory, showing that universal encoders, trained with Wasserstein loss regularization, achieve distortion performance comparable to task-specific models. These findings highlight the practicality and effectiveness of universal representations for multi-task lossy compression.

%===========================================================================================
\nocite{liu2019classification}
\bibliographystyle{plainnat}
\bibliography{main}

%===========================================================================================
\newpage
\appendix

\begin{center}
    \Large\bfseries
    Universal Rate-Distortion-Classification
    Representations \\ for Lossy Compression
\end{center}

%\vspace{0.1em}

\begin{center}
    \large\bfseries
    Supplementary Material
\end{center}

\vspace{2em}

%===========================================================================================
\section{Theoretical Results}
\subsection{Proof of Theorem \ref{TheoremRDCGS}}\label{Appendix_Proof_RDC_GS}

\begin{reptheorem}{TheoremRDCGS}[Information Rate-Distortion-Classification Function for Gaussian Source]
~\cite{Wang2024} Let \( X\sim \mathcal{N}(\mu_X,\sigma_X^2) \) be a Gaussian source and \( S\sim \mathcal{N}(\mu_S,\sigma_S^2) \) be an associated classification variable, with a covariance of \( \text{Cov}(X,S) = \theta_1 \). The problem is feasible if $C \geq \frac{1}{2} \log\left(1 - \frac{\theta_1^2}{\sigma_S^2 \sigma_X^2}\right) + h(S)$. For the mean squared error distortion (i.e., \( \mathbb{E}[\Delta(X, \hat{X})] = \mathbb{E}[(X-\hat{X})^2] \)), the information rate-distortion-classification function is achieved by a jointly Gaussian estimator \( \hat{X} \) and is given by
  \begin{align*}
    R(D,C) =
    \begin{cases}
        \frac{1}{2} \log \frac{\sigma_X^2}{D}, \quad
        & D \leq \sigma_X^2 \left( 1 - \frac{1}{\rho^2} \left( 1 - e^{-2h(S) + 2C} \right) \right), \\
        -\frac{1}{2} \log \left(1 - \frac{1}{\rho^2} \left( 1 - e^{-2h(S) + 2C} \right) \right) \quad
        & D > \sigma_X^2 \left( 1 - \frac{1}{\rho^2} \left( 1 - e^{-2h(S) + 2C} \right) \right) \\
        0, \quad 
        & C > h(S) \text{ and } D > \sigma_X^2.
    \end{cases}
  \end{align*}
where \( \rho = \frac{\theta_1}{\sigma_S \sigma_X} \) represents the correlation coefficient between \( X \) and \( S \), while \( h(\cdot) \) denotes the differential entropy of a continuous random variable.
\end{reptheorem}

\begin{proof}
The proof is an extension of the argument presented in~\cite{Wang2024}, with additional details provided to ensure completeness. Consider the $R(D,C)$ problem with MSE distortion as follows
\begin{mini!}|s|[2] % mini! = minimize
{p_{\hat{X}|X}} % optimization variable
{I(X;\hat{X})} % objective function
{} % label for optimization problem
{R(D,C) =} % optimization result} % optimization result
\addConstraint{\mathbb{E}[(X-\hat{X})^2]}{\leq D}{} % constraint 1
\addConstraint{H(S | \hat{X})}{\leq C.}{} % constraint 3
\end{mini!}
where $(X,S)$ are jointly Gaussian variables with covariance $\text{Cov}(X,S)=\theta_1$. 

We begin with a lemma from estimation theory that compares the performance of Gaussian and non-Gaussian estimators with matched second-order statistics.
\begin{lemma}
\label{lem:guassian_var}
\cite{wornell2005inference} Let \( \hat{X} \) be a random variable with mean \( \mathbb{E}[\hat{X}] = \mu_{\hat{X}} \), variance \( \mathrm{Var}(\hat{X}) = \sigma_{\hat{X}}^2 \), and covariance \( \mathrm{Cov}(X, \hat{X}) = \theta_2 \). Let \( \hat{X}_G \) be jointly Gaussian with \( X \) and share the same mean, variance, and covariance as \( \hat{X} \). Then,
\begin{equation*}
    \mathbb{E}[(X - \mathbb{E}[X | \hat{X}_G])^2] 
    \geq 
    \mathbb{E}[(X - \mathbb{E}[X | \hat{X}])^2].
\end{equation*}
\end{lemma}

This result implies that the MMSE of a general (possibly non-Gaussian) estimator \( \hat{X} \) is always less than or equal to that of a Gaussian estimator with the same first and second-order moments.

We next show that restricting \( \hat{X} \) to be jointly Gaussian with \( X \) incurs no loss of generality. Since \( \hat{X} \) and its jointly Gaussian counterpart \( \hat{X}_G \) are matched in mean, variance, and covariance with \( X \), it follows that $\mathbb{E}[(X - \hat{X})^2] = \mathbb{E}[(X - \hat{X}_G)^2]$. The following lemma will be used to formalize this result.
\begin{lemma}
\label{lem:guassian_var_mutual_information}
Given $\mu_{\hat{X}}$, $\sigma_{\hat{X}}^2$, $\theta_2$, and $(S \to X \to \hat{X})$, we have that
\begin{equation*}
    H(S|\hat{X}) \geq H(S|\hat{X}_G).
\end{equation*}
\end{lemma}

\begin{proof} 
By hypothesis, \(\hat{X}\) and \(\hat{X}_G\) have the same mean $\mu_{\hat{X}}$, the same variance \(\sigma_{\hat{X}}^2\), and the same covariance with \(X\), say $\text{Cov}(X, \hat{X}) = \text{Cov}(X, \hat{X}_G) = \theta_2$.

It notes that,
\begin{align}
    H(S \mid \hat{X}) &= H(S) - I(S; \hat{X}), \\
    H(S \mid \hat{X}_G) &= H(S) - I(S; \hat{X}_G),
\end{align}

From information theory, it is a standard result that if we fix the covariances among \((S, \hat{X})\), the joint Gaussian distribution yields the largest possible \(I(S; \hat{X})\). Equivalently, for any non-Gaussian \(\hat{X}\) with the same second-order statistics,
\begin{equation}
    I(S; \hat{X}) \leq I(S; \hat{X}_G),
\end{equation}

Therefore, we rewrite
\begin{equation}
\begin{split}
    H(S) - I(S; \hat{X}) 
    &\geq H(S) - I(S; \hat{X}_G),\\
    H(S \mid \hat{X}) &\geq H(S \mid \hat{X}_G),
\end{split}
\end{equation}
This completes the proof.
\end{proof}

Following the derivation approach in~\cite{UniversalRDPs}, we show that the mutual information \( I(X; \hat{X}) \) is minimized when \( \hat{X} \) is constrained to be jointly Gaussian with \( X \). Specifically, we have:
\begin{equation}
\begin{aligned}
    I(X ; \hat{X}) &= h(X) - h(X | \hat{X}) \\
    &\geq h(X) - h(X - \mathbb{E}[X | \hat{X}]) \\
    &\labrel\geq{relctr:gaussian_entropy} h(X) - \frac{1}{2} \log \left(2\pi e\, \mathbb{E}[(X - \mathbb{E}[X | \hat{X}])^2]\right) \\
    &\labrel\geq{relctr:gaussian_var} h(X) - \frac{1}{2} \log \left(2\pi e\, \mathbb{E}[(X - \mathbb{E}[X | \hat{X}_G])^2]\right) \\
    &= h(X) - h(X - \mathbb{E}[X | \hat{X}_G]) \\
    &\labrel={relctr:est_error} h(X) - h(X | \hat{X}_G) \\
    &= I(X ; \hat{X}_G),
\end{aligned}
\end{equation}
where inequality~(\ref{relctr:gaussian_entropy}) follows from the fact that the Gaussian distribution maximizes differential entropy for a given variance; inequality~(\ref{relctr:gaussian_var}) follows from Lemma~\ref{lem:guassian_var}; and equality~(\ref{relctr:est_error}) holds because the estimation error is independent of \( \hat{X}_G \). 

Therefore, it suffices to solve the following optimization problem:
\begin{mini!}|s|[2]
    {p_{\hat{X}_G|X}} 
    {I(X; \hat{X}_G)} 
    {} 
    {R(D, C) =} 
    \addConstraint{\mathbb{E}[(X - \hat{X}_G)^2]}{\leq D}
    \addConstraint{H(S | \hat{X}_G)}{\leq C.}
\end{mini!}

By applying the closed-form expressions for differential entropy and mutual information of jointly Gaussian variables~\cite[Chapter~8]{cover1999elements}, we obtain:
\begin{equation}
I(X;\hat{X}) = -\frac{1}{2}\log\left(1 - \frac{\theta_2^2}{\sigma_X^2 \sigma_{\hat{X}}^2}\right),
\end{equation}
and for the classification constraint:
\begin{equation*}
\begin{split}  
H(S|\hat{X}) &= h(S) - I(S; \hat{X}) \leq C, \\
I(S; \hat{X}) &\geq h(S) - C, \\
-\frac{1}{2} \log\left(1 - \frac{\theta_1^2}{\sigma_S^2 \sigma_X^4} \times \frac{\theta_2^2}{\sigma_{\hat{X}}^2}\right) &\geq h(S) - C.
\end{split}
\end{equation*}

Additionally, the mean squared error between \( X \) and \( \hat{X} \) can be expressed as:
\begin{equation}
\mathbb{E}[(X - \hat{X})^2] = (\mu_X - \mu_{\hat{X}})^2 + \sigma_X^2 + \sigma_{\hat{X}}^2 - 2 \theta_2.
\end{equation}

Then, the RDC problem can be represented by the following optimization formulation:
\begin{mini!}|s|[2] % mini! = minimize
{\mu_{\hat{X}},\sigma_{\hat{X}},\theta_2} % optimization variable
{-\frac{1}{2}\log(1-\frac{\theta_2^2}{\sigma_X^2\sigma_{\hat{X}}^2})} % objective function
{\label{RDCGS}} % label for optimization problem
{R(D,C) =} % optimization result} % optimization result
\addConstraint{(\mu_X-\mu_{\hat{X}})^2 +\sigma_X^2+\sigma_{\hat{X}}^2-2\theta_2 \leq D}{\label{RDCGS-D}} % constraint 1
\addConstraint{-\frac{1}{2}\log(1-\frac{\theta_1^2}{\sigma_S^2\sigma_X^4} \frac{\theta_2^2}{\sigma_{\hat{X}}^2})}{\geq h(S) - C.}{\label{RDCGS-C}} % constraint 3
\end{mini!}

The first observation is that, without loss of generality, we may assume \( \mu_{\hat{X}} = \mu_X \). The objective function~\eqref{RDCGS} and the classification constraint~\eqref{RDCGS-C} are invariant to the mean of \( \hat{X} \), while the distortion constraint~\eqref{RDCGS-D} is minimized when the means are matched. The second observation concerns feasibility. The classification constraint~\eqref{RDCGS-C} becomes infeasible when $C < \frac{1}{2} \log\left(1 - \frac{\theta_1^2}{\sigma_S^2 \sigma_X^2} \right) + h(S)$. 

To ensure the mutual information in~\eqref{RDCGS} is well-defined, we require $ 1 - \frac{\theta_2^2}{\sigma_X^2 \sigma_{\hat{X}}^2} > 0,$ which implies $\frac{\theta_2^2}{\sigma_{\hat{X}}^2} < \sigma_X^2$. Then, the mutual information between \( S \) and the jointly Gaussian variable \( \hat{X}_G \) is upper bounded by
\begin{align*}
    I(S ; \hat{X}_G) 
    &= -\frac{1}{2} \log\left(1 - \frac{\theta_1^2}{\sigma_S^2 \sigma_X^4} \times \frac{\theta_2^2}{\sigma_{\hat{X}}^2} \right), \\
    &\leq -\frac{1}{2} \log\left(1 - \frac{\theta_1^2}{\sigma_S^2 \sigma_X^2} \right),
\end{align*}
which makes~\eqref{RDCGS-C} infeasible if $C < \frac{1}{2} \log\left(1 - \frac{\theta_1^2}{\sigma_S^2 \sigma_X^2} \right) + h(S)$. Thus, we assume that 
\begin{align*}
    C \geq \frac{1}{2} \log\left(1 - \frac{\theta_1^2}{\sigma_S^2 \sigma_X^2} \right) + h(S).
\end{align*}

\textbf{Case 1:} The distortion constraint~\eqref{RDCGS-D} is active, while the classification constraint~\eqref{RDCGS-C} is inactive.

According to Shannon’s rate-distortion theory for Gaussian sources with squared-error distortion~\cite{cover1999elements}, the rate-distortion function for \( X \sim \mathcal{N}(\mu_X, \sigma_X^2) \) is given by
\begin{equation}\label{eqn:rd_gaussian}
    R(D) =
    \begin{cases}
        \frac{1}{2} \log \left( \frac{\sigma_X^2}{D} \right), & 0 \leq D \leq \sigma_X^2, \\
        0, & D > \sigma_X^2,
    \end{cases}
\end{equation}
where the optimal distribution is achieved by a conditional \( p_{\hat{X}|X} \) such that \( \hat{X} \sim \mathcal{N}(\mu_X, \sigma_X^2 - D) \) when \( D \leq \sigma_X^2 \). In this regime, the optimal rate is \( \frac{1}{2} \log \left( \frac{\sigma_X^2}{D} \right) \), achieved by choosing \( \sigma_{\hat{X}}^2 = \sigma_X^2 - D \) and \( \theta_2 = \sigma_X^2 - D \).

The classification constraint~\eqref{RDCGS-C} is inactive if
\begin{align*}
    I(S ; \hat{X}_G) = -\frac{1}{2} \log\left(1 - \frac{\theta_1^2 (\sigma_X^2 - D)}{\sigma_S^2 \sigma_X^4} \right) \geq h(S) - C,
\end{align*}
which is equivalent to the distortion satisfying $ D \leq \sigma_X^2 - \frac{\sigma_S^2 \sigma_X^4}{\theta_1^2} \left( 1 - e^{-2h(S) + 2C} \right).$

\textbf{Case 2:} The distortion constraint~\eqref{RDCGS-D} is inactive, while the classification constraint~\eqref{RDCGS-C} is active.

When $D > \sigma_X^2 - \frac{\sigma_S^2 \sigma_X^4}{\theta_1^2} \left(1 - e^{2C - 2h(S)}\right)$, the distortion constraint is no longer active. In this case, the classification constraint~\eqref{RDCGS-C} imposes the following lower bound:
\begin{align}
    \frac{\theta_2^2}{\sigma_{\hat{X}}^2} \geq \frac{\sigma_S^2 \sigma_X^4}{\theta_1^2} \left(1 - e^{-2h(S) + 2C}\right). \label{incorporatingC}
\end{align}

Since the mutual information expression~\eqref{RDCGS} is increasing in \( \frac{\theta_2^2}{\sigma_{\hat{X}}^2} \), incorporating the bound in~\eqref{incorporatingC} yields the following lower bound on the rate:
\begin{align*}
    I(X; \hat{X}_G) \geq -\frac{1}{2} \log \left(1 - \frac{\sigma_S^2 \sigma_X^2}{\theta_1^2} \left(1 - e^{-2h(S) + 2C}\right)\right).
\end{align*}

Equality is attained when
\[
\frac{\theta_2^2}{\sigma_{\hat{X}}^2} = \frac{\sigma_S^2 \sigma_X^4}{\theta_1^2} \left(1 - e^{-2h(S) + 2C}\right),
\]
i.e., when the classification constraint is active. Without loss of generality, we choose
\[
\sigma_{\hat{X}}^2 = \theta_2 = \frac{\sigma_S^2 \sigma_X^4}{\theta_1^2} \left(1 - e^{-2h(S) + 2C}\right),
\]
and substitute into the distortion expression:
\begin{align*}
    \mathbb{E}[(X - \hat{X}_G)^2] 
    &= \sigma_X^2 + \sigma_{\hat{X}}^2 - 2 \theta_2 \\
    &= \sigma_X^2 - \frac{\sigma_S^2 \sigma_X^4}{\theta_1^2} \left(1 - e^{-2h(S) + 2C}\right) < D,
\end{align*}
which confirms that the distortion constraint is also satisfied.

\textbf{Case 3:} Both the distortion constraint~\eqref{RDCGS-D} and the classification constraint~\eqref{RDCGS-C} are inactive.

When \( C > h(S) \) and \( D > \sigma_X^2 \), we may set \( \sigma_{\hat{X}} = 0 \), i.e., let \( \hat{X}_G \) be a constant. This yields \( I(X ; \hat{X}_G) = 0 \), and both constraints are trivially satisfied.

\noindent In summary, the information rate-distortion-classification function under MSE distortion is given by:
\begin{align*}
R(D, C) =
\begin{cases}
\frac{1}{2} \log \left( \frac{\sigma_X^2}{D} \right), & D \leq \sigma_X^2 \left( 1 - \frac{1}{\rho^2} (1 - e^{-2h(S) + 2C}) \right), \\
-\frac{1}{2} \log \left( 1 - \frac{1}{\rho^2} (1 - e^{-2h(S) + 2C}) \right), & D > \sigma_X^2 \left( 1 - \frac{1}{\rho^2} (1 - e^{-2h(S) + 2C}) \right), \\
0, & C > h(S),\; D > \sigma_X^2,
\end{cases}
\end{align*}
where \( \rho = \frac{\theta_1}{\sigma_S \sigma_X} \) denotes the correlation coefficient between \( X \) and \( S \).

\end{proof}

%===========================================================================================
\subsection{Proof of Theorem \ref{TheoremDCR_GS}}\label{Appendix_Proof_DCR_GS}

\begin{reptheorem}{TheoremDCR_GS}[Information Distortion-Rate-Classification Function for Gaussian Source]
Consider a Gaussian source \( X\sim \mathcal{N}(\mu_X,\sigma_X^2) \) and an associated classification variable \( S\sim \mathcal{N}(\mu_S,\sigma_S^2) \) with covariance \( \text{Cov}(X,S) = \theta_1 \). The problem is feasible if the classification loss satisfies $ C \geq \frac{1}{2} \log\left(1 - \frac{\theta_1^2}{\sigma_S^2 \sigma_X^2}\right) + h(S)$. Under the mean squared error distortion, the information distortion-classification-rate function is given by
\begin{align*}
    D(C, R) = 
\begin{cases} 
    \sigma_X^2 e^{-2R}, \\ 
    \hspace{2cm} C \geq \frac{1}{2} \log\left(1 - \frac{\theta_1^2 (\sigma_X^2 - \sigma_X^2 e^{-2R})}{\sigma_S^2 \sigma_X^4} \right) + h(S) \\\\
    \sigma_X^2 + \frac{\sigma_S^2 \sigma_X^4 (1 - e^{2C - 2h(S)})}{\theta_1^2} 
    - 2 \frac{\sigma_S \sigma_X^3 \sqrt{(1 - e^{-2h(S) + 2C})(1 - e^{-2R})}}{\theta_1}, \\ 
    \hspace{2cm} \frac{1}{2} \log\left(1 - \frac{\theta_1^2}{\sigma_S^2 \sigma_X^2}\right) + h(S) 
    \leq C < \frac{1}{2} \log\left(1 - \frac{\theta_1^2 (\sigma_X^2 - \sigma_X^2 e^{-2R})}{\sigma_S^2 \sigma_X^4} \right) + h(S)\\
    \sigma_X^2, \\
    \hspace{2cm} C > h(S).
\end{cases}
\end{align*}
\end{reptheorem}

\begin{proof}
Consider the distortion-classification-rate function \( D(C, R) \) under the MSE distortion criterion as follows
\begin{mini!}|s|[2] % mini! = minimize
{p_{\hat{X}|X}} % optimization variable
{\mathbb{E}[(X-\hat{X})^2]} % objective function
{} % label for optimization problem
{D(C,R) =} % optimization result} % optimization result
\addConstraint{I(X;\hat{X})}{\leq R}{} % constraint 1
\addConstraint{H(S | \hat{X})}{\leq C.}{} % constraint 3
\end{mini!}
where \( (X, S) \) are jointly Gaussian random variables with covariance \( \mathrm{Cov}(X, S) = \theta_1 \). As in the proof of Theorem~\ref{TheoremRDCGS}, the optimal solution is attained when \( \hat{X} \) is also Gaussian and jointly distributed with \( X \) \cite{Wang2024}. Thus, the optimization reduces to a parameter search over the mean \( \mu_{\hat{X}} \), variance \( \sigma_{\hat{X}}^2 \), and covariance \( \mathrm{Cov}(X, \hat{X}) = \theta_2 \). 

Then, the $D(C,R)$ problem can be formulated as:
\begin{mini!}|s|[2] % mini! = minimize
{\mu_{\hat{X}},\sigma_{\hat{X}},\theta_2} % optimization variable
{\!\!\!\!\!(\mu_X-\mu_{\hat{X}})^2 +\sigma_X^2+\sigma_{\hat{X}}^2-2\theta_2} % objective function
{\label{DCR}} % label for optimization problem
{D(C,R) \! = \!\!\!} % optimization result} % optimization result
\addConstraint{\!\!\!\!\!\!\!\!\!\!\!\!\!\!\!\! -\frac{1}{2}\log\left(1 - \frac{\theta_2^2}{\sigma_X^2 \sigma_{\hat{X}}^2}\right) \leq R}{\label{DCR_I}} % constraint 1
\addConstraint{\!\!\!\!\!\!\!\!\!\!\!\!\!\!\!\! -\frac{1}{2}\log\left( 1-\frac{\theta_1^2}{\sigma_S^2\sigma_X^4} \frac{\theta_2^2}{\sigma_{\hat{X}}^2} \right)}{\geq h(S) - C.}{\label{DCR_C}} % constraint 3
\end{mini!}

Without loss of generality, we assume \( \mu_{\hat{X}} = \mu_X \). To ensure that the mutual information expression in~\eqref{DCR_I} is well-defined, it is necessary that \( 1 - \frac{\theta_2^2}{\sigma_X^2 \sigma_{\hat{X}}^2} > 0 \), i.e., \( \frac{\theta_2^2}{\sigma_{\hat{X}}^2} < \sigma_X^2 \). Under this condition, the mutual information between \( S \) and \( \hat{X} \) is upper bounded as
\begin{align*}
\begin{split}
I(S; \hat{X}) &= -\frac{1}{2}\log\left(1 - \frac{\theta_1^2}{\sigma_S^2 \sigma_X^4} \times \frac{\theta_2^2}{\sigma_{\hat{X}}^2} \right),\\ 
&\leq -\frac{1}{2} \log\left(1 - \frac{\theta_1^2}{\sigma_S^2 \sigma_X^2} \right),   
\end{split}
\end{align*}
which implies that constraint~\eqref{DCR_C} becomes infeasible if \( C < \frac{1}{2} \log\left(1 - \frac{\theta_1^2}{\sigma_S^2 \sigma_X^2}\right) + h(S) \). Therefore, to guarantee feasibility, we assume throughout that
\begin{equation*}
C \geq \frac{1}{2} \log\left(1 - \frac{\theta_1^2}{\sigma_S^2 \sigma_X^2} \right) + h(S).
\end{equation*}

\textbf{Case 1.} Constraint~(\ref{DCR_I}) is active and constraint~(\ref{DCR_C}) is inactive.

Recall the classical Shannon distortion-rate function for a Gaussian source \( X \sim \mathcal{N}(\mu_X, \sigma_X^2) \)~\cite{cover1999elements}:
\begin{equation*}\label{eqn:rd_gaussian}
    R(D)={\begin{cases}{\frac {1}{2}}\log (\frac{\sigma _{X}^{2}}{D}),& 0\leq D\leq \sigma _{X}^{2}\\
    0,& D>\sigma _{X}^{2}.\end{cases}}
\end{equation*}
And $D(R)=\sigma_X^2 e^{-2R}$ with the optimal solution attained by some \( p_{\hat{X}|X} \) where \( \hat{X} \sim \mathcal{N}(\mu_X, \sigma_X^2 - D) \). In this case, we have \( D(C, R) = \sigma_X^2 e^{-2R} \), and constraint~(\ref{DCR_I}) is active. Then, \( \theta_2 = \sigma_X^2 - D = \sigma_X^2(1 - e^{-2R}) \).

To verify that constraint~(\ref{DCR_C}) is not active, we require:
\begin{align*}
I(S,\hat{X}) &=-\frac{1}{2}\log\left( 1-\frac{\theta_1^2(\sigma_X^2-D)}{\sigma_S^2\sigma_X^4} \right)\geq h(S)-C,
\end{align*}
This condition is equivalent to
\begin{equation*}
D \leq \sigma_X^2 - \frac{\sigma_S^2 \sigma_X^4}{\theta_1^2} \left(1 - e^{-2h(S) + 2C} \right),
\end{equation*}
which characterizes the maximal allowable distortion that still satisfies the classification constraint.

Hence,
\begin{align*}
\begin{split}
C & \geq \frac{1}{2}\log\left( 1-\frac{\theta_1^2(\sigma_X^2-D)}{\sigma_S^2\sigma_X^4} \right) + h(S), \\
&= \frac{1}{2}\log\left( 1-\frac{\theta_1^2 (\sigma_X^2 - \sigma_X^2 e^{-2 R})}{\sigma_S^2\sigma_X^4} \right) +h(S),
\end{split}
\end{align*}

\textbf{Case 2.}  Constraint~(\ref{DCR_C}) is active and constraint~(\ref{DCR_I}) is inactive.

This corresponds to the range: 
\begin{equation*}
\begin{split}
&\frac{1}{2} \log\left(1 - \frac{\theta_1^2}{\sigma_S^2 \sigma_X^2}\right) + h(S) \leq C < \frac{1}{2}\log(1-\frac{\theta_1^2 (\sigma_X^2 - \sigma_X^2 e^{-2 R})}{\sigma_S^2\sigma_X^4}) +h(S),   
\end{split}
\end{equation*}

When \( D > \sigma_X^2 - \frac{\sigma_S^2 \sigma_X^4}{\theta_1^2}(1 - e^{2C - 2h(S)}) \), the rate constraint~\eqref{DCR_I} becomes inactive. In this case, the classification constraint~\eqref{DCR_C} imposes a lower bound on the quantity \( \frac{\theta_2^2}{\sigma_{\hat{X}}^2} \), given by
\begin{align}
\frac{\theta_2^2}{\sigma_{\hat{X}}^2} \geq \frac{\sigma_S^2 \sigma_X^4}{\theta_1^2} \left(1 - e^{-2h(S) + 2C} \right).
\label{incorporatingC}
\end{align}

Since the mutual information \( I(X; \hat{X}) \) is a monotonically increasing function of \( \frac{\theta_2^2}{\sigma_{\hat{X}}^2} \), incorporating the bound in~\eqref{incorporatingC} yields the following inequality:
\begin{align*}
I(X; \hat{X}) \geq -\frac{1}{2} \log\left(1 - \frac{\sigma_S^2 \sigma_X^2}{\theta_1^2} \left(1 - e^{-2h(S) + 2C} \right) \right),
\end{align*}
with equality attained when 
\[
\frac{\theta_2^2}{\sigma_{\hat{X}}^2} = \frac{\sigma_S^2 \sigma_X^4}{\theta_1^2} (1 - e^{2C - 2h(S)}),
\]
i.e., when the classification constraint~\eqref{DCR_C} is active. Without loss of generality, when this constraint is active, the optimal decoder satisfies:
\begin{equation*}
\sigma_{\hat{X}}^2 = \theta_2 = \frac{\sigma_S^2 \sigma_X^4 (1 - e^{-2h(S) + 2C})}{\theta_1^2}.
\end{equation*}

Since constraint~(\ref{DCR_I}) is inactive, we must ensure:
\begin{equation*}
I(X; \hat{X}) = -\frac{1}{2}\log(1-\frac{\theta_2^2}{\sigma_X^2\sigma_{\hat{X}}^2}) 
\leq R,  
\end{equation*}
which is equivalent to $\theta_2^2 \leq \sigma_X^2 \sigma_{\hat{X}}^2 (1 - e^{-2 R})$ or $\theta_2 \leq\sigma_X \sigma_{\hat{X}} \sqrt{ 1 - e^{-2 R}}$. Substituting into the distortion expression:
\begin{equation*}
\begin{split}
D(C,R) &= \sigma_X^2 + \sigma_{\hat{X}}^2 - 2\theta_2,\\
&\geq \sigma_X^2 + \frac{\sigma_S^2\sigma_X^4(1-e^{-2h(S)+2C})}{\theta_1^2} - 2\frac{\sigma_S^2\sigma_X^3 \sqrt{(1-e^{-2h(S)+2C})(1 - e^{-2 R})}}{\theta_1^2},
\end{split}
\end{equation*}

\textbf{Case 3.}  Both constraint \eqref{DCR_I} and constraint \eqref{DCR_C} are inactive.

When \( C > h(S) \) and \( D(C, R) \geq \sigma_X^2 \), a trivial solution can be constructed by setting \( \sigma_{\hat{X}} = 0 \), i.e., letting \( \hat{X} \) be a constant. Under this setting, we have \( I(X; \hat{X}) = 0 \leq R \), and both constraints are satisfied. 

In summary, combining the three cases, the closed-form expression for the information distortion-classification-rate function \( D(C, R) \) under MSE distortion is given by
\begin{align*}
D(C, R) = 
\begin{cases} 
    \sigma_X^2 e^{-2R}, \\ 
    \hspace{2cm} C \geq \frac{1}{2} \log\left(1 - \frac{\theta_1^2 (\sigma_X^2 - \sigma_X^2 e^{-2R})}{\sigma_S^2 \sigma_X^4} \right) + h(S) \\\\
    \sigma_X^2 + \frac{\sigma_S^2 \sigma_X^4 (1 - e^{2C - 2h(S)})}{\theta_1^2} 
    - 2 \frac{\sigma_S \sigma_X^3 \sqrt{(1 - e^{-2h(S) + 2C})(1 - e^{-2R})}}{\theta_1}, \\ 
    \hspace{2cm} \frac{1}{2} \log\left(1 - \frac{\theta_1^2}{\sigma_S^2 \sigma_X^2}\right) + h(S) 
    \leq C < \frac{1}{2} \log\left(1 - \frac{\theta_1^2 (\sigma_X^2 - \sigma_X^2 e^{-2R})}{\sigma_S^2 \sigma_X^4} \right) + h(S)\\
    \sigma_X^2, \\
    \hspace{2cm} C > h(S).
\end{cases}
\end{align*}
\end{proof}

%===========================================================================================
\subsection{Proof of Theorem \ref{Theorem_DCR_Convexity}}\label{Appendix_Proof_DCR_Convexity}
\begin{reptheorem}{Theorem_DCR_Convexity}
The function \( D(C, R) \), defined for all points \( (C, R) \) where \( D(C, R) < +\infty \), satisfies the following properties:
\begin{itemize}
    \item It is monotonically non-increasing in both \( C \) and \( R \);
    \item It is convex.
\end{itemize}
\end{reptheorem}
\begin{proof}
The proof idea follows the result in~\cite{Wang2024}. Let $\Delta_n$ denote the probability simplex in $\mathbb{R}^n$. For a vector $\mathbf{q}$, let $\mathbf{q}[i]$ denote its $i$-th entry. For a matrix $\mathbf{T}$, let $\mathbf{T}[ji]$ denote the entry in the $j$-th row and $i$-th column. Following the geometric approach in \cite{Witsenhausen_1975}, we interpret the distortion-classification-rate function $D(C,R)$ and prove its convexity.

Consider discrete random variables $(X, S)$ where the marginal distribution of $S$ is given by $\mathbf{Tq}$, with $\mathbf{q} \in \Delta_n$ and $\mathbf{q}[i] = P(X = i)$. The transition matrix $\mathbf{T} \in \mathbb{R}^{m \times n}$ is defined by $\mathbf{T}[ji] = P(S = j | X = i)$.

Let $\mathbf{w} \in \Delta_k$ with $\mathbf{w}[\alpha] = P(\hat{X} = \alpha)$. Define the transition matrix $\mathbf{B} = (\mathbf{p}_1, \cdots, \mathbf{p}_k)$ with each $\mathbf{p}_\alpha \in \Delta_n$. This induces a random variable $X'$ with distribution $\mathbf{p} = \mathbf{B} \mathbf{w} = \sum_{\alpha=1}^k \mathbf{w}[\alpha] \mathbf{p}_\alpha$, and $S' = \mathbf{T} \mathbf{p} = \sum_{\alpha=1}^k \mathbf{w}[\alpha] \mathbf{T} \mathbf{p}_\alpha$.

For a distortion function $\Delta(x, \hat{x})$, define $d_{i, \alpha} = \Delta(i, \alpha)$, and set $\mathbf{d}_\alpha = (d_{1, \alpha}, \cdots, d_{n, \alpha})^T$. The following quantities can then be computed:
\begin{align}
  \mathbf{p} &= \sum_{\alpha=1}^k \mathbf{w}[\alpha] \mathbf{p}_\alpha, \label{map_p} \\
  \rho &= \sum_{\alpha=1}^k \mathbf{w}[\alpha] H(\mathbf{p}_\alpha), \label{map_R} \\
  \eta &= \sum_{\alpha=1}^k \mathbf{w}[\alpha] H(\mathbf{T} \mathbf{p}_\alpha), \label{map_C} \\
  \delta &= \sum_{\alpha=1}^k \mathbf{w}[\alpha] \mathbf{d}_\alpha^T \mathbf{p}_\alpha. \label{map_D}
\end{align}

Define the set $\mathcal{S} = \{ (\mathbf{p}_\alpha, H(\mathbf{p}_\alpha), H(\mathbf{T} \mathbf{p}_\alpha), \mathbf{d}_\alpha^T \mathbf{p}_\alpha) : \mathbf{p}_\alpha \in \Delta_n \}$ and let $\mathcal{C}$ be its convex hull.

\begin{lemma}
The set of all tuples $(\mathbf{p}, \rho, \eta, \delta)$ obtained from \eqref{map_p}--\eqref{map_D} on all $\mathbf{w} \in \Delta_k$ and $\mathbf{p}_\alpha \in \Delta_n$ is exactly the convex hull $\mathcal{C}$.
\end{lemma}
\begin{proof}
By construction, \eqref{map_p}--\eqref{map_D} define $(\mathbf{p}, \rho, \eta, \delta)$ as a convex combination of points in $\mathcal{S}$. In contrast, any convex combination of points in $\mathcal{S}$ can be expressed in this form.
\end{proof}

We now consider the DCR optimization problem:
\begin{mini!}|s|
  {p_{\hat{X}|X}}
  {\mathbb{E}[\Delta(X, \hat{X})]}
  {\label{DCR}}
  {D(C, R) =}
  \addConstraint{I(X; \hat{X})}{\leq R}
  \addConstraint{H(S | \hat{X})}{\leq C.}
\end{mini!}

Given fixed marginal $\mathbf{q}$, the above is equivalent to:
\begin{align*}
  D(C, R) = \inf \{ \delta : (\mathbf{q}, \delta, R, C) \in \mathcal{C} \}.
\end{align*}

\textbf{Convexity:} For any pairs $\{(D_1, C_1), (D_2, C_2)\}$ where $D(C, R) < \infty$, let
\begin{align*}
  \delta_1 &= D(C_1, R_1), \, \delta_2 = D(C_2, R_2), \\
  (C_\lambda, R_\lambda) &= \lambda(C_1, R_1) + (1 - \lambda)(C_2, R_2), \\
  \delta_\lambda &= \lambda \delta_1 + (1 - \lambda) \delta_2.
\end{align*}

Since $\mathcal{C}$ is convex, it follows that $D(C_\lambda, R_\lambda) \leq \delta_\lambda = \lambda D(C_1, R_1) + (1 - \lambda) D(C_2, R_2)$, proving convexity.

\textbf{Monotonicity:} Since $D(C, R)$ is convex, it suffices to show that $D(C, R) \geq D(H(S|X), R_{\max})$ for all $0 \leq C \leq H(S|X)$ and $0 \leq R \leq R_{\max}$, where $R_{\max} \triangleq \max_\alpha \mathbb{E}[I(X; \alpha)]$.

By choosing $\hat{X} = \alpha^* = \arg\max_\alpha \mathbb{E}[I(X; \alpha)]$, we obtain $\mathbb{E}[\Delta(X, \hat{X})] = 0$, $H(S|\hat{X}) = H(S|X)$, and $I(X; \hat{X}) = R_{\max}$, yielding $D(H(S|X), R_{\max}) = 0$. Thus,
\begin{align*}
  D(C, R) \geq 0 = D(H(S|X), R_{\max}).
\end{align*}

By convexity, for any $\lambda \in [0,1]$, we have:
\begin{align*}
  D(\lambda C + (1 - \lambda) H(S|X), \lambda R + (1 - \lambda) R_{\max})
  \leq \lambda D(C, R).
\end{align*}
This confirms the monotonicity of $D(C, R)$.
\end{proof}

%===========================================================================================
\subsection{Proof of Theorem \ref{Theorem_CDR_Bound}}\label{Appendix_Proof_CDR_Bound}
\begin{reptheorem}{Theorem_CDR_Bound}[Bound of Classification-Distortion-Rate Function]
For a given rate \( R(D, C) \) and mean squared error distortion, the following inequality holds:
\begin{equation}
D_{\max} \leq 2 D_{\min},
\end{equation}
where \( D_{\min} \) and \( D_{\max} \) are defined in (\ref{D_min}) and (\ref{D_max}), respectively.

This bound is reached by an estimator \( \hat{X} \), characterized by the conditional distribution \( p_{\hat{X}|X} \), achieving the minimum classification loss \( H(S | \hat{X}) = C_{\min} \) while incurring a maximum mean squared error distortion of exactly \( 2 D_{\min} \).
\end{reptheorem}

\begin{proof} 
Recall that the estimator \( \hat{X}_{D,C} \) is defined by a conditional distribution \( p_{\hat{X}_{D,C}|X} \) that satisfies the classification constraint \( H(S \mid \hat{X}_{D,C}) = C_{\min} \). By construction, the random variables \( X \), \( M \), and \( \hat{X}_{D,C} \) form a Markov chain \( X \to M \to \hat{X}_{D,C} \). Applying the law of total expectation, we compute:
\begin{align} \label{eq:XtEhatX}
\mathbb{E}[X^\top \hat{X}_{D,C}] 
&= \mathbb{E}\left[ \mathbb{E}[X^\top \hat{X}_{D,C} | M] \right] \nonumber \\
&= \mathbb{E}\left[ \mathbb{E}[X | M]^\top \mathbb{E}[\hat{X}_{D,C} | M] \right] \nonumber \\
&= \mathbb{E}\left[ \| \mathbb{E}[X | M] \|^2 \right],
\end{align}
where the last equality holds since \( X \) and \( \hat{X}_{D,C} \) are conditionally independent and identically distributed given \( M \), and thus share the same conditional mean.

Similarly, the second moment of \( \hat{X}_{D,C} \) satisfies:
\begin{align} \label{eq:EhatX}
\mathbb{E}[\|\hat{X}_{D,C}\|^2] 
= \mathbb{E}\left[ \mathbb{E}[\|\hat{X}_{D,C}\|^2 | M] \right] 
= \mathbb{E}\left[ \mathbb{E}[\|X\|^2 | M] \right] 
= \mathbb{E}[\|X\|^2],
\end{align}
by the law of total expectation and the symmetry in conditional distributions.

Using these results, the mean squared error of \( \hat{X}_{D,C} \) can be expressed as:
\begin{align}
\mathbb{E}[\| X - \hat{X}_{D,C} \|^2] 
&= \mathbb{E}[\|X\|^2] - 2 \mathbb{E}[X^\top \hat{X}_{D,C}] + \mathbb{E}[\|\hat{X}_{D,C}\|^2] \nonumber \\
&= 2\left( \mathbb{E}[\|X\|^2] - \mathbb{E}[ \| \mathbb{E}[X | M] \|^2] \right) \nonumber \\
&= 2\, \mathbb{E}[\| X - \mathbb{E}[X | M] \|^2] \nonumber \\
&= 2\, \mathbb{E}[\| X - \hat{X}_{\text{MMSE}} \|^2],
\end{align}
where the second equality follows from~\eqref{eq:XtEhatX} and~\eqref{eq:EhatX}, and the third from the orthogonality principle.

Hence, \( \hat{X}_{D,C} \) is a distribution-preserving estimator whose MSE is exactly twice that of the MMSE estimator. It follows that
\begin{equation}
D_{\max} \leq \mathbb{E}[\| X - \hat{X}_{D,C} \|^2] = 2 D_{\min},
\end{equation}
thereby completing the proof.
\end{proof}

%===========================================================================================
\subsection{Proof of Theorem \ref{Theorem_gaussian_universality}}\label{Appendix_Proof_GS_Universality}

\begin{reptheorem}{Theorem_gaussian_universality}[No Rate Penalty for a Gaussian Source]
Let \( X \sim \mathcal{N}(\mu_X, \sigma_X^2) \) be a scalar Gaussian source and let \( S \sim \mathcal{N}(\mu_S, \sigma_S^2) \) be a classification variable with covariance \( \mathrm{Cov}(X, S) = \theta_1 \). Assume that the distortion is measured using the mean squared error and that the classification loss is measured via conditional entropy. Let \( \Theta \) denote any non-empty set of distortion-classification constraint pairs \( (D, C) \). Then,
\begin{equation}
    A(\Theta) = 0,
\end{equation}
which implies that satisfying the most demanding constraint in \( \Theta \) is sufficient to simultaneously satisfy all others using a fixed encoder and there is no rate penalty for universality in this case.

Furthermore, consider any representation \( Z \) that is jointly Gaussian with \( X \) and satisfies
\begin{equation}
    I(X; Z) = \sup_{(D, C) \in \Theta} R(D, C).
\end{equation}
Then the following inclusion holds:
\begin{equation}\label{eqn:sup_rate}
    \Theta \subseteq \Omega(p_{Z|X}) = \Omega(I(X; Z)),
\end{equation}
meaning that \( Z \) achieves the maximal distortion-classification region at rate \( I(X; Z) \); i.e., all constraints in \( \Theta \) are simultaneously achievable via appropriate decoders applied to a common representation \( Z \).
\end{reptheorem}

\begin{proof}
The proof method follows the approach presented in~\cite{UniversalRDPs}. Let \( R = \sup_{(D,C) \in \Theta} R(D,C) \). By definition, \( \Theta \subseteq \Omega(R) \), where \( \Omega(R) \) denotes the set of all achievable distortion-classification pairs at rate \( R \). The lower boundary of this region, the optimal tradeoff curve, is characterized by 
\begin{align*}
\label{DC_LowerBoundary}
D=\sigma_X^2 + \frac{\sigma_S^2\sigma_X^4(1-e^{-2h(S)+2C})}{\theta_1^2} - 2\frac{\sigma_S^2\sigma_X^3 \sqrt{(1-e^{-2h(S)+2C})(1 - e^{-2 R})}}{\theta_1^2},\\
C\in\left[ \frac{1}{2} \log\left(1 - \frac{\theta_1^2}{\sigma_S^2 \sigma_X^2}\right) + h(S), \frac{1}{2}\log(1-\frac{\theta_1^2 (\sigma_X^2 - \sigma_X^2 e^{-2 R})}{\sigma_S^2\sigma_X^4}) +h(S)  \right).\nonumber
\end{align*}

Every point in \( \Omega(R) \) is component-wise dominated by a point on this boundary. Consider a representation \( Z \) that is jointly Gaussian with \( X \), such that \( I(X; Z) = R \). This implies the squared correlation between \( X \) and \( Z \) satisfies \( \rho_{XZ}^2 = 1 - 2^{-2R} \), where 
\begin{equation*}
    \rho_{XZ}= \frac{\text{Cov}(X,Z)}{\sigma_X\sigma_Z} = \frac{\mathbb{E}[(X-\mu_X)(Z-\mu_Z)]}{\sigma_X\sigma_Z}.
\end{equation*}

For any point \( (D, C) \) on the boundary, define the corresponding reconstruction as:
\begin{equation*}
    \hat{X}_{D,C} = \mbox{sign}(\rho_{XZ})\gamma (Z-\mu_Z)+\mu_X,
\end{equation*}
where $\mbox{sign}(\rho_{XZ}) = \begin{cases}
1, &\text{for } \rho_{XZ}\geq 0,\\
-1, &\text{for } \rho_{XZ} < 0.
\end{cases}$

With this construction, we have:
\begin{equation*}
\begin{split}
    \mathbb{E}[\hat{X}_{D,C}] &= \mathbb{E}[\mbox{sign}(\rho_{XZ})\gamma (Z-\mu_Z)+\mu_X], \\
    &= \mbox{sign}(\rho_{XZ}) \gamma (\mathbb{E}[Z] - \mathbb{E}[Z]) + \mu_X, \\
    &= \mathbb{E}[X],
\end{split}
\end{equation*}

And,
\begin{equation*}
    \text{Var}(\hat{X}_{D,C}) = \gamma^2 \text{Var}(Z),
\end{equation*}
\begin{equation*}
    \text{Cov}(X, \hat{X}_{D,C})  = \gamma \text{Cov}(X, Z),
\end{equation*}

We now choose \( \text{Var}(\hat{X}_{D,C}) \) such that constraint~(\ref{DCR_I}) is inactive and constraint~(\ref{DCR_C}) is active:
\begin{equation*}
\text{Var}(\hat{X}_{D,C}) = \frac{\sigma_S^2 \sigma_X^4 (1 - e^{-2 h(S) + 2C})}{\theta_1^2},
\end{equation*}

Solving for \( \gamma \), we obtain: $\gamma = \frac{\sigma_S \sigma_X^2 \sqrt{1-e^{-2 h(S) + 2C}}}{\theta_1 \sigma_Z}$.

We now verify that this reconstruction satisfies the distortion and classification constraints.

\textbf{Distortion constraint.}
\begin{align*}
\mathbb{E}[\| X-\hat{X}_{D,C} \|^2] &=\sigma^2_X+\sigma^2_{\hat{X}_{D,C}}-2\theta_2,\\
&=\sigma^2_X + \sigma^2_{\hat{X}_{D,C}} - 2\gamma \text{Cov}(X,Z),\\
&=\sigma^2_X  + \sigma^2_{\hat{X}_{D,C}} - 2 \frac{\sigma^2_{\hat{X}_{D,C}}}{\sigma_Z} \sigma_X \sigma_Z \sqrt{1 - e^{-2R}},\\
&= \sigma^2_X + \sigma^2_{\hat{X}_{D,C}} - 2 \sigma^2_{\hat{X}_{D,C}} \sigma_X \sqrt{1 - e^{-2R}},\\
&= \sigma^2_X + \frac{\sigma_S^2\sigma_X^4(1-e^{-2h(S)+2C})}{\theta_1^2} - 2 \frac{\sigma_S \sigma_X^3 \sqrt{(1 - e^{-2h(S) + 2C})(1 - e^{-2R})}}{\theta_1},\\
&= D.
\end{align*} 

\textbf{Classification constraint.}
\begin{align*}
H(S|\hat{X}_{D,C}) &= h(S) - I(S|\hat{X}_{D,C}), \\
&= h(S) + \frac{1}{2} \log\left( 1 - \frac{\theta_1^2}{\sigma_S^2 \sigma_X^4 } \frac{\theta_2^2}{\sigma_{\hat{X}_{D,C}}^2} \right), 
\end{align*}
where  
\begin{equation*}
    \frac{\theta_2^2}{\sigma_{\hat{X}_{D,C}}^2} = \frac{\sigma_S^2 \sigma_X^4 (1 - e^{-2 h(S) + 2C})}{\theta_1^2},
\end{equation*}

Substituting in, we get:
\begin{align*}
H(S|\hat{X}_{D,C}) &= h(S) + \frac{1}{2} \log\left( 1 - \frac{\theta_1^2}{\sigma_S^2 \sigma_X^4 } \frac{\sigma_S^2 \sigma_X^4 (1 - e^{-2 h(S) + 2C})}{\theta_1^2} \right), \\
&= h(S) + \frac{1}{2} \log \left( e^{-2h(S) + 2C} \right),\\
&= C. 
\end{align*}

This confirms that for any given \( \text{Cov}(X, Z) \), one can always choose a scalar \( \gamma \) to generate \( \hat{X}_{D,C} \) that satisfies both distortion and classification constraints. That is, every point \( (D, C) \in \Theta \) can be realized by applying an appropriate decoder to a fixed Gaussian representation \( Z \) of \( X \) such that \( I(X; Z) = \sup_{(D, C) \in \Theta} R(D, C) \). Therefore, \( \Omega(p_{Z|X}) = \Omega(R) \), which implies that the rate penalty is zero: $A(\Theta) = I(X; Z) - R = 0$.
\end{proof}

%=========================================================================================== 
\subsection{Proof of Proposition \ref{Prop_equivalence}}\label{Appendix_Proof_prop_equivalence}

\begin{repproposition}{Prop_equivalence}[Equivalence of zero rate penalty and full distortion-classification region]
Assume the following regularity conditions are satisfied:
\begin{itemize}
    \item \( \sup_{(D, C) \in \Omega(R)} R(D, C) = R' \);
    \item The infimum in the definition of \( R(\Omega(R')) \) is achieved.
\end{itemize}
Then, the rate penalty satisfies \( A(\Omega(R')) = 0 \) if and only if there exists a representation \( Z \) with \( I(X; Z) = R' \) such that $\Omega(p_{Z|X}) = \Omega(I(X; Z))$.
\end{repproposition}

\begin{proof}
Suppose there exists a representation \( Z \) satisfying \( I(X; Z) = R' \) and \( \Omega(p_{Z|X}) = \Omega(I(X; Z)) \). Then, by definition, \( R(\Omega(R')) \leq R' \). Under Condition 1), we know that the rate penalty satisfies \( A(\Omega(R')) \leq 0 \). Since the rate penalty is always non-negative, it must be that \( A(\Omega(R')) = 0 \).

Now consider Condition 2), which guarantees the existence of a representation \( Z \) such that \( I(X; Z) = R(\Omega(R')) \) and \( \Omega(p_{Z|X}) \supseteq \Omega(R') \). If \( A(\Omega(R')) = 0 \), then it follows that
\[
R(\Omega(R')) = \sup_{(D, C) \in \Omega(R')} R(D, C).
\]
Combined with Condition 1), this implies \( R(\Omega(R')) = R' \). Furthermore, since \( I(X; Z) = R' \), we have \( \Omega(p_{Z|X}) \subseteq \Omega(R') \). Together with the previous inclusion, we conclude that \( \Omega(p_{Z|X}) = \Omega(R') \).
\end{proof}

%===========================================================================================
\subsection{Proof of Theorem \ref{Theorem_general_universality}} \label{Appendix_Proof_General_Universality}

\begin{reptheorem}{Theorem_general_universality}
[Universality Achievable Region for a General Source]
Consider a general source \( X \sim p_X \) and a classification variable \( S \), with covariance \( \mathrm{Cov}(X, S) = \theta_1 \). Assume that distortion is measured using mean squared error, and classification loss is measured using conditional entropy, $H(S | \hat{X})$. Let \( Z \) be an arbitrary representation of \( X \), and define the minimum mean square error (MMSE) estimator of \( X \) given \( Z \) as \( \tilde{X} = \mathbb{E}[X | Z] \). Then, the closure of this region, denoted \( \mathrm{cl}\big(\Omega(p_{Z \mid X})\big) \), satisfies
\begin{equation*}
\Omega(p_{Z|X}) \subseteq 
\left\{ (D, C) : 
    D \geq \mathbb{E}[\|X - \tilde{X}\|^2] + 
    \begin{aligned}
        &\inf_{p_{\hat{X}}} \quad W^2_2(p_{\tilde{X}}, p_{\hat{X}}) \\
        &\text{s.t} \quad H(S | \hat{X}) \leq C
    \end{aligned}
\right\}
\subseteq \mathrm{cl}(\Omega(p_{Z|X})),
\end{equation*}
where \( W^2_2(\cdot, \cdot) \) denotes the squared 2-Wasserstein distance is defined as  
\begin{equation*}
\begin{aligned}
 W_{2}^2(p_X, p_{\hat{X}}) = \underset{p_{X,\hat{X}}} \inf \,\, &\mathbb{E}[\|X - \hat{X}\|^2] \\
 \text{s.t.} \,\,\,\,\,\,\,\,\,\, & \int_{-\infty }^{\infty} p_{X,\hat{X}}dX = p_{\hat{X}}, \int_{-\infty }^{\infty} p_{X,\hat{X}}d\hat{X} = p_{X}.
\end{aligned}
\end{equation*} 

Furthermore, the closure \( \mathrm{cl}(\Omega(p_{Z|X})) \) contains the following two extreme points:
\begin{align*}
(D^{(a)}, C^{(a)}) &= \left( \mathbb{E}[\|X-\tilde{X}\|^2], \sum_{s}\sum_{\tilde{x}} p_{\tilde{X}} p_{S|\tilde{X}} \log\frac{1}{p_{S|\tilde{X}}} \right), \\
(D^{(b)}, C^{(b)}) &= \left( \mathbb{E}[\|X - \tilde{X}\|^2] + W_2^2(p_{\tilde{X}}, p_{\hat{X}}^{(U)}), C_{\min}^{(\text{U})} \right),
\end{align*}
where \( C_{\min}^{(\text{U})} \) is the minimum classification loss achievable at a given distortion level under universality, defined by the following optimization problem. With the given rate, $R(\Theta)$, we have:
\begin{mini!}|s|[2]
    {p_{\hat{X}}} 
    {H(S | \hat{X})} 
    {\label{C_min}} 
    {p_{\hat{X}}^{C_{\min}^{(\text{U})}} = \arg} 
    \addConstraint{\mathbb{E}[\|X - \hat{X}\|^2]}{\leq D.}
\end{mini!}
The corresponding minimum classification loss is given by
\begin{equation*}
    C_{\min}^{(\text{U})} = \sum_{s}\sum_{\hat{x}^{(U)}} p_{\hat{X}}^{(U)} p_{S|\hat{X}^{(U)}} \log\frac{1}{p_{S|\hat{X}^{(U)}}}.
\end{equation*}
\end{reptheorem}

\begin{proof}
The proof approach is inspired by the main techniques introduced in~\cite{UniversalRDPs}. For any \( (D, C) \in \Omega(p_{Z|X}) \), there exists a reconstruction variable \( \hat{X}_{D,C} \) jointly distributed with \( (X, Z) \), such that the Markov chain \( X \leftrightarrow Z \leftrightarrow \hat{X}_{D,C} \) holds, the distortion satisfies \( \mathbb{E}[\Delta(X, \hat{X}_{D,C})] \leq D \), and the classification uncertainty is bounded by \( H(S | \hat{X}_{D,C}) \leq C \). 

Since $\tilde{X} = \mathbb{E}[X|Z]$ is the orthogonal projection (in the $L_2$-sense) of $X$ onto any $\sigma$-algebra generated by $Z$, one has the well-known Pythagorean identity:
\begin{equation*}
\mathbb{E}[\|X - \hat{X}_{D,C}\|^2] = \mathbb{E}[\|X - \tilde{X}\|^2] + \mathbb{E}[\|\tilde{X} - \hat{X}_{D,C}\|^2].   
\end{equation*}

Concretely, the cross-term $\mathbb{E}[(X - \tilde{X}) (\tilde{X} - \hat{X}_{D,C})] = 0$
vanishes because $\tilde{X}$ is the conditional mean (the orthogonal projection argument). Then,
\begin{align*}
D\geq\mathbb{E}[\|X-\hat{X}_{D,C}\|^2]  =\mathbb{E}[\|X-\tilde{X}\|^2]+\mathbb{E}[\|\tilde{X}-\hat{X}_{D,C}\|^2].
\end{align*}

Now consider the Wasserstein-2 distance between the marginals \( p_{\tilde{X}} \) and \( p_{\hat{X}_{D,C}} \), which is defined as:
\[
W_2^2(p_{\tilde{X}}, p_{\hat{X}_{D,C}}) = \inf_{p_{\tilde{X}', \hat{X}'}} \mathbb{E}[\|\tilde{X}' - \hat{X}'\|^2],
\]
where \( \tilde{X}' \sim p_{\tilde{X}} \) and \( \hat{X}' \sim p_{\hat{X}_{D,C}} \). Since \( (\tilde{X}, \hat{X}_{D,C}) \) is one feasible coupling of these marginals,
\[
\mathbb{E}[\|\tilde{X} - \hat{X}_{D,C}\|^2] \geq W_2^2(p_{\tilde{X}}, p_{\hat{X}_{D,C}}).
\]

Therefore,
\[
D \geq \mathbb{E}[\|X - \tilde{X}\|^2] + W_2^2(p_{\tilde{X}}, p_{\hat{X}_{D,C}}).
\]

Since \( H(S | \hat{X}_{D,C}) \leq C \), the marginal distribution \( p_{\hat{X}_{D,C}} \) belongs to the constraint set \( \{ p_{\hat{X}} : H(S | \hat{X}) \leq C \} \). We have, $p_{\hat{X}_{D,C}}$ is one feasible distribution of the set $\left\{ p_{\hat{X}}:\;
\begin{aligned}
    &\inf_{p_{\hat{X}}}   W^2_2(p_{\tilde{X}},p_{\hat{X}}) \\
    &\text{s.t. } H(S|\hat{X}) \leq C
\end{aligned}
\right\}$, then
\[
W_2^2(p_{\tilde{X}}, p_{\hat{X}_{D,C}}) \geq 
\left\{
\begin{aligned}
&\inf_{p_{\hat{X}}} W_2^2(p_{\tilde{X}}, p_{\hat{X}}) \\
&\text{s.t. } H(S | \hat{X}) \leq C
\end{aligned}
\right.
\]

This leads to the outer bound:
\begin{equation*}
\Omega(p_{Z|X}) \! \subseteq \! \left\{ \! (D,C) \! : \! D  \geq  \mathbb{E}[{\|X-\tilde{X}\|^2}] \!+\! \begin{aligned}
    &\inf_{p_{\hat{X}}}  W^2_2(p_{\tilde{X}},p_{\hat{X}}) \\
    &\text{s.t. } H(S|\hat{X}) \leq C
\end{aligned} \! \right\}
\end{equation*}

Now, to show the approximate tightness of this bound, let \( (D', C') \) be any point in the above region. For any \( \epsilon > 0 \), there exists a distribution \( p_{\hat{X}'} \) such that:
\[
H(S |\hat{X}') \leq C', \quad D' + \epsilon \geq \mathbb{E}[\|X - \tilde{X}\|^2] + W_2^2(p_{\tilde{X}}, p_{\hat{X}'}).
\]

By the Markov condition, we can construct a random variable \( \hat{X}' \) such that \( X \to Z \to \hat{X}' \), and
\[
\mathbb{E}[\|\tilde{X} - \hat{X}'\|^2] \leq W_2^2(p_{\tilde{X}}, p_{\hat{X}'}) + \epsilon.
\]
Thus,
\[
\mathbb{E}[\|X - \hat{X}'\|^2] = \mathbb{E}[\|X - \tilde{X}\|^2] + \mathbb{E}[\|\tilde{X} - \hat{X}'\|^2] \leq D' + 2\epsilon.
\]

It follows that:
\begin{equation*}
\begin{split}
    \Omega(p_{Z|X}) \! &\subseteq \! \left\{ (D,C) : D \geq \mathbb{E}{\|X-\tilde{X}\|^2} \! + \! \begin{aligned}
    &\inf_{p_{\hat{X}}}   W^2_2(p_{\tilde{X}},p_{\hat{X}}) \\
    &\text{s.t. } H(S|\hat{X}) \leq C
\end{aligned} \right\}  \!  \subseteq \! \mbox{cl}(\Omega(p_{Z|X})).
\end{split}
\end{equation*}

Now consider the characterization of conditional entropy:
\begin{align*}
    H(S|\hat{X})\! &=\!\! \sum_{s}\sum_{\hat{x}} p_{S,\hat{X}}\log\frac{1}{p_{S|\hat{X}}},\\
    &=\!\! \sum_{s}\sum_{\hat{x}} p_{\hat{X}} p_{S|\hat{X}}\log\frac{1}{p_{S|\hat{X}}}.
\end{align*}

By choosing \( \hat{X} = \tilde{X} \), it follows that \( p_{\hat{X}} = p_{\tilde{X}} \), which yields:
\begin{equation*}
     H(S|\hat{X}) = H(S|\tilde{X}) = \sum_{s}\sum_{\tilde{x}} p_{\tilde{X}} p_{S|\tilde{X}} \log\frac{1}{p_{S|\tilde{X}}},
\end{equation*}
and define:
\begin{equation*}
(D^{(a)},C^{(a)})=\left( \mathbb{E}[\|X-\tilde{X}\|^2], \sum_{s}\sum_{\tilde{x}} p_{\tilde{X}} p_{S|\tilde{X}} \log\frac{1}{p_{S|\tilde{X}}} \right).   
\end{equation*}

Next, for a given rate \( R(\Theta) \), define $\hat{X}^{(U)} \sim p_{\hat{X}}^{(U)}$:
\begin{mini!}|s|[2] % mini! = minimize
{p_{\hat{X}}} % optimization variable
{ H(S|\hat{X})} % objective function
{\label{C_min_Appendix}} % label for optimization problem
{ p_{\hat{X}}^{(U)} = \arg} % optimization result} % optimization result
\addConstraint{\mathbb{E}[\| X-\hat{X} \|^2]}{\leq D.} % constraint 1
%\addConstraint{I(X;Z) = R(\Theta).}% constraint 3
\end{mini!}
and let:
\begin{equation*}
    C_{\min}^{(U)} = \sum_{s}\sum_{\hat{x}^{(U)}} p_{\hat{X}}^{(U)} p_{S|\hat{X}^{(U)}} \log\frac{1}{p_{S|\hat{X}^{(U)}}}.
\end{equation*}

Let $\hat{X} = \hat{X}^{(U)}$ implies \( p_{\hat{X}} = p_{\hat{X}}^{(U)} \), and define:
\[
(D^{(b)}, C^{(b)}) = \left( \mathbb{E}[\|X - \tilde{X}\|^2] + W_2^2(p_{\tilde{X}}, p_{\hat{X}}^{(U)}), C_{\min}^{(\text{U})} \right).
\]
Thus, by selecting \( p_{\hat{X}} = p_{\tilde{X}} \) and \( p_{\hat{X}} = p_{\hat{X}}^{(U)} \), we confirm that both \( (D^{(a)}, C^{(a)}) \) and \( (D^{(b)}, C^{(b)}) \) lie in this region: 
\[
\left\{ (D,C) : D \geq \mathbb{E}{\|X-\tilde{X}\|^2} + \begin{aligned}
    &\inf_{p_{\hat{X}}}  W^2_2(p_{\tilde{X}},p_{\hat{X}}) \\
    &\text{s.t. } H(S|\hat{X}) \leq C
\end{aligned} \right\}.
\]
\end{proof}

%===========================================================================================
\subsection{Proof of Theorem \ref{Theorem_Quantitative_Results}}\label{Appendix_Proof_Quantitative_Results}

\begin{reptheorem}{Theorem_Quantitative_Results}[Quantitative Characterization for General Sources]
Let \( \hat{X}_{D_1, C_1} \) be the optimal reconstruction associated with the point \( (D_1, C_1) \) on the rate-distortion-classification trade-off curve, i.e., in the conventional sense where $
I(X; \hat{X}_{D_1, C_1}) = R(D_1, C_1)$.
Then, the upper-left extreme point of the achievable region \( \Omega(p_{\hat{X}_{D_1, C_1} \mid X}) \) coincides with that of \( \Omega(R) \); specifically, $
(D^{(a)}, C^{(a)}) = (D_1, C_1)$. Consider the lower-right extreme points \( (D^{(b)}, C^{(b)}) \in \Omega(p_{\hat{X}_{D_1, C_1} \mid X}) \) and \( (D_3, C_3) \in \Omega(R) \), where \( C_3 = C_{\min}^{\text{(Conv)}} \) and \( R(D_3, C_3) = R(D_1, \infty) \). The distortion gap between these two points satisfies the following inequality:
\begin{equation}
D_3 - D^{(b)} \geq \sigma_X^2 + \sigma_{\hat{X}_{D_3, C_3}}^2 
- 2 \sigma_{\hat{X}_{D_3, C_3}} \sqrt{\sigma_X^2 - D_1} - 2D_1,
\end{equation}
and the corresponding distortion ratio is lower bounded as:
\begin{equation}
\frac{D_3}{D^{(b)}} \geq 
\frac{\sigma_X^2 + \sigma_{\hat{X}_{D_3, C_3}}^2 
- 2 \sigma_{\hat{X}_{D_3, C_3}} \sqrt{\sigma_X^2 - D_1}}{2D_1}.
\end{equation}

Moreover, in the special case where the squared 2-Wasserstein distance between \( p_X \) and \( p_{\hat{X}_{D_3, C_3}} \) is zero, i.e., \( W_2^2(p_X, p_{\hat{X}_{D_3, C_3}}) = 0 \), which implies \( \sigma_X^2 = \sigma_{\hat{X}_{D_3, C_3}}^2 \), the distortion gap becomes negligible under the following conditions:
\begin{equation}
D_3 - D^{(b)} \approx 0 \quad \text{when } D_1 \approx 0 \text{ or } D_1 \approx \sigma_X^2,
\end{equation}
\begin{equation}
\frac{D_3}{D^{(b)}} \approx 1 \quad \text{when } D_1 \approx \sigma_X^2.
\end{equation}
\end{reptheorem}

\begin{proof} 
The proof idea follows the result in~\cite{UniversalRDPs}. We begin by noting that \( C_3 = C_{\min}^{\text{(Conv)}} \), and from the hypothesis \( R(D_3, C_3) = R(D_1, \infty) \), we define $\hat{X}^{(\text{Conv})} \sim p_{\hat{X}}^{\text{(Conv)}}$:
\begin{mini!}|s|[2] % mini! = minimize
{p_{\hat{X}}} % optimization variable
{ H(S|\hat{X})} % objective function
{} % label for optimization problem
{ p_{\hat{X}}^{\text{(Conv)}} = \arg} % optimization result} % optimization result
\addConstraint{\mathbb{E}[\| X-\hat{X} \|^2]}{\leq D.} % constraint 1
%\addConstraint{I(X;\hat{X}) = R(D,C).}% constraint 3
\end{mini!}
The corresponding minimum classification uncertainty is given by:
\begin{equation*}
    C_{\min}^{\text{(Conv)}} = \sum_{s}\sum_{\hat{x}^{\text{(Conv)}}} p_{\hat{X}}^{\text{(Conv)}} p_{S|\hat{X}^{\text{(Conv)}}} \log\frac{1}{p_{S|\hat{X}^{\text{(Conv)}}}},
\end{equation*}

Next, observe that:
\begin{align*}
D_3&=\mathbb{E}[\|X-\hat{X}_{D_3,C_3}\|^2], \\
&= \sigma^2_X + \sigma_{\hat{X}_{D_3, C_3}}^2 -2 \text{Cov}(X, \hat{X}_{D_3,C_3}), \\
&=\sigma^2_X + \sigma_{\hat{X}_{D_3, C_3}}^2 -2\mathbb{E}[(X-\mu_X)^T(\hat{X}_{D_3,C_3}-\mu_{\hat{X}_{D_3, C_3}})],
\end{align*} 

Thus, we have:
\begin{equation*}
\label{Eq_D3}
\mathbb{E}[(X-\mu_X)^T(\hat{X}_{D_3,C_3}-\mu_{\hat{X}_{D_3, C_3}}) = \frac{\sigma^2_X + \sigma_{\hat{X}_{D_3, C_3}}^2 - D_3}{2}.
\end{equation*}

We now utilize the inequality \( I(X; \mathbb{E}[X | \hat{X}_{D_3,C_3}]) \leq I(X; \hat{X}_{D_3,C_3}) = R(D_1, \infty) \), which implies: $\mathbb{E}[\|X-\mathbb{E}[X|\hat{X}_{D_3,C_3}]\|^2]\geq D_1$. 

Observe that $\hat{X}_{D_3,C_3} - \mu_{\hat{X}_{D_3, C_3}}$ has a certain correlation with $X - \mu_X$, so we can use a linear predictor idea to get a simpler upper bound. By the orthogonality principle:
\begin{equation*}
\mathbb{E}[\|X  -  \mathbb{E}[X | \hat{X}_{D_3,C_3}]\|^2] \! \leq \! \mathbb{E}[\|X - \mu_X - c (\hat{X}_{D_3,C_3} - \mu_{\hat{X}_{D_3, C_3}})\|^2].
\end{equation*}
The inequality says the best predictor $\mathbb{E}[X | \hat{X}_{D_3,C_3}]$ is no worse (in MSE) than any fixed linear predictor of the form $\mu_X + c (\hat{X}_{D_3,C_3} - \mu_{\hat{X}_{D_3, C_3}})$, where
\begin{equation*}
\begin{split}
c \! &= \!\! \frac{\text{Cov}(X-\mu_X, \hat{X}_{D_3,C_3} - \mu_{\hat{X}_{D_3, C_3}})}{\text{Var}({\hat{X}_{D_3,C_3} - \mu_{\hat{X}_{D_3, C_3}}})},\\
& \!\! = \!\! \frac{\mathbb{E}[(X-\mu_X)^T(\hat{X}_{D_3,C_3}-\mu_{\hat{X}_{D_3, C_3}})]}{\sigma_{\hat{X}_{D_3, C_3}}^2} = \frac{\sigma^2_X + \sigma_{\hat{X}_{D_3, C_3}}^2 - D_3}{2\sigma_{\hat{X}_{D_3, C_3}}^2},
\end{split}
\end{equation*}

Therefore, the MSE of the optimal conditional expectation is at most the MSE of this linear estimator.  
\begin{align*}
D_1&\leq\mathbb{E}[\|X-\mathbb{E}[X|\hat{X}_{D_3,C_3}]\|^2],\\
&\leq\mathbb{E}[\|X-\mu_X-c(\hat{X}_{D_3,C_3}-\mu_{\hat{X}_{D_3, C_3}})\|^2],\\
&= \mathbb{E}[\|X - \mu_X\|^2] + c^2 \mathbb{E}[\|\hat{X}_{D_3,P_3} - \mu_{\hat{X}_{D_3, C_3}}\|^2] - 2c \mathbb{E}[(X - \mu_X)^T (\hat{X}_{D_3,P_3} - \mu_{\hat{X}_{D_3, C_3}})],\\
&=\sigma^2_X + c^2 \sigma_{\hat{X}_{D_3, C_3}}^2 \!\!\!\! - 2c\mathbb{E}[(X-\mu_X)^T(\hat{X}_{D_3,C_3} \! - \! \mu_{\hat{X}_{D_3, C_3}})],\\
&=\sigma^2_X + \left( \frac{\sigma^2_X + \sigma_{\hat{X}_{D_3, C_3}}^2 - D_3}{2\sigma_{\hat{X}_{D_3, C_3}}^2} \right)^2 \sigma_{\hat{X}_{D_3, C_3}}^2 - 2 \left( \frac{\sigma^2_X + \sigma_{\hat{X}_{D_3, C_3}}^2 - D_3}{2\sigma_{\hat{X}_{D_3, C_3}}^2} \right) \left( \frac{\sigma^2_X + \sigma_{\hat{X}_{D_3, C_3}}^2 - D_3}{2} \right),\\
&= \sigma^2_X - \frac{(\sigma^2_X + \sigma_{\hat{X}_{D_3, C_3}}^2 - D_3)^2}{4 \sigma_{\hat{X}_{D_3, C_3}}^2},
\end{align*}

Rearranging terms yields:
\begin{equation*}
(\sigma^2_X + \sigma_{\hat{X}_{D_3, C_3}}^2 - D_3)^2 \leq 4 \sigma_{\hat{X}_{D_3, C_3}}^2 (\sigma^2_X - D_1),   
\end{equation*}

Under the assumptions \( \sigma_X^2 + \sigma_{\hat{X}_{D_3,C_3}}^2 - D_3 \geq 0 \) and \( \sigma_X^2 - D_1 \geq 0 \), it follows that:
\begin{equation*}
\begin{split}
\sigma^2_X + \sigma_{\hat{X}_{D_3, C_3}}^2 - D_3 \leq 2 \sigma_{\hat{X}_{D_3, C_3}} \sqrt{\sigma_X^2 - D_1},\\
D_3 \geq \sigma^2_X + \sigma_{\hat{X}_{D_3, C_3}}^2 - 2 \sigma_{\hat{X}_{D_3, C_3}} \sqrt{\sigma_X^2 - D_1},
\end{split}
\end{equation*}

From Theorem \ref{Theorem_CDR_Bound}, we have:
\begin{equation*}
D^{(b)} \leq D_3 \leq 2D_1,
\end{equation*}
Hence, 
\begin{align*}
&D_3 - D^{(b)}\geq \sigma^2_X + \sigma_{\hat{X}_{D_3, C_3}}^2 - 2 \sigma_{\hat{X}_{D_3, C_3}} \sqrt{\sigma_X^2 - D_1} - 2D_1,\\
&\frac{D_3}{D^{(b)}}\geq \frac{\sigma^2_X + \sigma_{\hat{X}_{D_3, C_3}}^2 - 2 \sigma_{\hat{X}_{D_3, C_3}} \sqrt{\sigma_X^2 - D_1}}{2D_1}.
\end{align*}

\underline{For $D_1 = 0$:}
\begin{align*}
D_3 - D^{(b)} &\geq \sigma^2_X + \sigma_{\hat{X}_{D_3, C_3}}^2 - 2 \sigma_{\hat{X}_{D_3, C_3}} \sigma_X,
\end{align*}

Moreover, in the case of $W_2^2(p_X, p_{\hat{X}_{D_3, C_3}}) = 0$ (i.e., $\sigma_X^2 = \sigma_{\hat{X}_{D_3, C_3}}^2$), we have:
\begin{equation*}
D_3 - D^{(b)}\stackrel{D_1\approx0 \hspace{1mm} \text{and} \hspace{1mm} \sigma_X^2 = \sigma_{\hat{X}_{D_3, C_3}}^2}{\approx}0, 
\end{equation*}

\underline{For $D_1 = \sigma_X^2$:}
\begin{align*}
&D_3 - D^{(b)} \geq \sigma^2_X + \sigma_{\hat{X}_{D_3, C_3}}^2 - 2\sigma^2_X,\\
&\frac{D_3}{D^{(b)}} \geq \frac{\sigma^2_X + \sigma_{\hat{X}_{D_3, C_3}}^2}{2 \sigma^2_X}.
\end{align*}

Again, in the case of $W_2^2(p_X, p_{\hat{X}_{D_3, C_3}}) = 0$, then:
\begin{align*}
&D_3 - D^{(b)}\stackrel{D_1\approx\sigma_X^2 \hspace{1mm} \text{and} \hspace{1mm} \sigma_X^2 = \sigma_{\hat{X}_{D_3, C_3}}^2}{\approx}0, \\
&\frac{D_3}{D^{(b)}}\stackrel{D_1\approx\sigma_X^2 \hspace{1mm} \text{and} \hspace{1mm} \sigma_X^2 = \sigma_{\hat{X}_{D_3, C_3}}^2}{\approx} 1.
\end{align*}

A similar method can be utilized to bound the upper-left corner of the curve, i.e., the gap between $(D_1,C_1)$ and the upper-left extreme point $(\tilde{D}^{(a)},\tilde{C}^{(a)})$ of the blue curve in Fig.~\ref{Fig:Universal_GernalSource_CDR}. Let \( \tilde{D}^{(a)} = \mathbb{E}[\|X - \mathbb{E}[X \mid \hat{X}_{D_3,C_3}]\|^2] \). Then:
\[
\tilde{D}^{(a)} \leq \sigma_X^2 - \frac{(\sigma_X^2 + \sigma_{\hat{X}_{D_3,C_3}}^2 - D_3)^2}{4 \sigma_{\hat{X}_{D_3,C_3}}^2},
\]
which, together with \( D_1 \geq \frac{1}{2} D_3 \), yields:
\begin{align*}
&\tilde{D}^{(a)}-D_1 \leq \sigma^2_X - \frac{(\sigma^2_X + \sigma_{\hat{X}_{D_3, C_3}}^2 - D_3)^2}{4 \sigma_{\hat{X}_{D_3, C_3}}^2} - \frac{D_3}{2},\\
&\frac{\tilde{D}^{(a)}}{D_1}\leq \frac{\sigma^2_X - \frac{(\sigma^2_X + \sigma_{\hat{X}_{D_3, C_3}}^2 - D_3)^2}{4 \sigma_{\hat{X}_{D_3, C_3}}^2}}{D_3/2}.
\end{align*}

\underline{For $D_3 = 0$:}
\begin{align*}
\tilde{D}^{(a)} - D_1 &\leq \sigma^2_X - \frac{(\sigma^2_X + \sigma_{\hat{X}_{D_3, C_3}}^2)^2}{4 \sigma_{\hat{X}_{D_3, C_3}}^2},
\end{align*}

If of $W_2^2(p_X, p_{\hat{X}_{D_3, C_3}}) = 0$, then:
\begin{equation*}
\tilde{D}^{(a)}-D_1 \stackrel{D_3\approx0 \hspace{1mm} \text{and} \hspace{1mm} \sigma_X^2 = \sigma_{\hat{X}_{D_3, C_3}}^2}{\approx}0, 
\end{equation*}

\underline{For $D_3 = 2\sigma_X^2$:}
\begin{align*}
&\tilde{D}^{(a)}-D_1 \leq \sigma^2_X - \frac{( \sigma_{\hat{X}_{D_3, C_3}}^2 - \sigma^2_X)^2}{4 \sigma_{\hat{X}_{D_3, C_3}}^2} - \sigma^2_X,\\
&\frac{\tilde{D}^{(a)}}{D_1} \leq \frac{\sigma^2_X - \frac{( \sigma_{\hat{X}_{D_3, C_3}}^2 - \sigma^2_X)^2}{4 \sigma_{\hat{X}_{D_3, C_3}}^2}}{\sigma_X^2}.
\end{align*}

Similarly, in the case of $W_2^2(p_X, p_{\hat{X}_{D_3, C_3}}) = 0$, we have:
\begin{align*}
&\tilde{D}^{(a)}-D_1 \stackrel{D_3\approx 2\sigma_X^2 \hspace{1mm} \text{and} \hspace{1mm} \sigma_X^2 = \sigma_{\hat{X}_{D_3, C_3}}^2}{\approx}0, \\
&\frac{\tilde{D}^{(a)}}{D_1} \stackrel{D_3\approx 2\sigma_X^2 \hspace{1mm} \text{and} \hspace{1mm} \sigma_X^2 = \sigma_{\hat{X}_{D_3, C_3}}^2}{\approx} 1.
\end{align*}
\end{proof}

%===========================================================================================
\subsection{Proof of Theorem \ref{theorem_asymptotic_universality}}\label{Appendix_Proof_asymptotic_universality}
\begin{reptheorem}{theorem_asymptotic_universality}
For any set \( \Theta \), the asymptotic universal rate-distortion-classification function satisfies
\begin{equation*}
    R^{(\infty)}(\Theta) = R(\Theta).
\end{equation*}
\end{reptheorem}

\begin{proof}
The proof approach follows the method presented in~\cite{UniversalRDPs}. Let \( Z \) be jointly distributed with \( X \) such that for every \( (D, C) \in \Theta \), there exists a decoder \( p_{\hat{X}_{D,C}|Z} \) satisfying the constraints \( \mathbb{E}[\Delta(X,\hat{X}_{D,C})] \leq D \) and \( H(S| \hat{X}_{D,C}) \leq C \). Define the product distribution $p_{Z^n|X^n} \triangleq \prod_{i=1}^n p_{Z(i)|X(i)}$, where each \( p_{Z(i)|X(i)} = p_{Z|X} \) for \( i = 1,\dots,n \).

By the strong functional representation lemma~\cite{li2018strong}, there exists a random variable \( V^{(n)} \), independent of \( X^n \), and a deterministic function \( \phi^{(n)} \) such that \( Z^n = \phi(X^n, V^{(n)}) \) and
\[
    H(\phi(X^n, V^{(n)}) | V^{(n)}) \leq I(X^n; Z^n) + \log(I(X^n; Z^n) + 1) + 4.
\]
With access to \( V^{(n)} \) at both encoder and decoder, one can use a class of prefix-free codes indexed by \( V^{(n)} \) to losslessly encode \( Z^n \) with expected codeword length no greater than \( I(X^n; Z^n) + \log(I(X^n; Z^n) + 1) + 5 \).

Furthermore, by the standard functional representation lemma, there exist a random variable \( V_{D,C} \), independent of \( (X^n, V^{(n)}) \), and a deterministic function \( \psi_{D,C} \) such that
\[
    p_{\psi_{D,C}(Z(i), V_{D,C})|Z(i)} = p_{\hat{X}_{D,C}|Z}.
\]
Both \( V^{(n)} \) and \( V_{D,C} \) can be generated from a shared random seed \( U^{(n)} \). Define
\[
    \hat{X}_{D,C}(i) = \psi_{D,C}(Z(i), V_{D,C}), \quad i = 1, \dots, n.
\]
Then it holds that
\begin{align*}
    &\frac{1}{n}\sum\limits_{i=1}^n\mathbb{E}[\Delta(X(i), \hat{X}_{D,C}(i))] = \mathbb{E}[\Delta(X, \hat{X}_{D,C})] \leq D, \\
    &H\left(S \Bigm| \frac{1}{n}\sum\limits_{i=1}^n \hat{X}_{D,C}(i)\right) = H(S | \hat{X}_{D,C}) \leq C.
\end{align*}

In addition, observe that
\begin{align*}
    &\frac{1}{n} I(X^n; Z^n) + \frac{1}{n} \log(I(X^n; Z^n) + 1) + \frac{5}{n} \\
    &= I(X; Z) + \frac{1}{n} \log(n I(X; Z) + 1) + \frac{5}{n} \\
    &\xrightarrow[n \to \infty]{} I(X; Z),
\end{align*}
implying that \( R^{(\infty)}(\Theta) \leq R(\Theta) \).

Now consider any sequence of random variables \( U^{(n)} \), encoding functions \( f^{(n)}_{U^{(n)}}: \mathcal{X}^n \to \mathcal{C}^{(n)}_{U^{(n)}} \), and decoding functions \( g^{(n)}_{U^{(n)},D,C}: \mathcal{C}^{(n)}_{U^{(n)}} \to \hat{\mathcal{X}}^n \) for each \( (D, C) \in \Theta \), satisfying the constraints in~(\ref{eq:constraintD}) and~(\ref{eq:constraintP}). Then,
\begin{align*}
\frac{1}{n}\mathbb{E}[\ell(f^{(n)}_{U^{(n)}}(X^n))] 
&\geq \frac{1}{n} H(f^{(n)}_{U^{(n)}}(X^n) | U^{(n)}) \\
&= \frac{1}{n} I(X^n; f^{(n)}_{U^{(n)}}(X^n) | U^{(n)}) \\
&= \frac{1}{n} \sum\limits_{i=1}^n I(X(i); f^{(n)}_{U^{(n)}}(X^n) | U^{(n)}, X^{i-1}) \\
&= \frac{1}{n} \sum\limits_{i=1}^n I(X(i); f^{(n)}_{U^{(n)}}(X^n), U^{(n)}, X^{i-1}) \\
&\geq \frac{1}{n} \sum\limits_{i=1}^n I(X(i); f^{(n)}_{U^{(n)}}(X^n), U^{(n)}) \\
&= I(X(T); f^{(n)}_{U^{(n)}}(X^n), U^{(n)} | T) \\
&= I(X(T); f^{(n)}_{U^{(n)}}(X^n), U^{(n)}, T),
\end{align*}
where \( T \) is uniformly distributed over \( \{1, \dots, n\} \) and is independent of both \( X^n \) and \( U^{(n)} \). Note that \( \hat{X}_{D,C}(T) \) is a function of \( (f^{(n)}_{U^{(n)}}(X^n), U^{(n)}, T) \) for each \( (D,C) \in \Theta \).

Because
\begin{align*}
    &p_{X(T)} = p_X, \\
    &\mathbb{E}[\Delta(X(T), \hat{X}_{D,C}(T))] = \frac{1}{n} \sum_{i=1}^n \mathbb{E}[\Delta(X(i), \hat{X}_{D,C}(i))] \leq D, \\
    &H(S | \hat{X}_{D,C}(T)) = H\left(S \Bigm| \frac{1}{n} \sum_{i=1}^n \hat{X}_{D,C}(i) \right) \leq C,
\end{align*}
it follows that
\[
    I(X(T); f^{(n)}_{U^{(n)}}(X^n), U^{(n)}, T) \geq R(\Theta).
\]
This completes the proof.
\end{proof}

% %===========================================================================================
\section{Experiments}

\subsection{Architecture}
The architecture employed in our experiments is a stochastic autoencoder combined with GAN and pre-trained classifier regularization. Each model is composed of an encoder, decoder, discriminator, and classifier. The detailed network architectures for these components are summarized in Table \ref{tab:architecture}, with each row indicating a sequence of layers.

\begin{table}[h]
    \centering
    \caption{Detailed architecture for encoder, decoder, discriminator, and pre-trained classifier used in MNIST experiments. l-ReLU denotes Leaky ReLU activation.}
    \label{tab:architecture}
    \begin{tabular}[t]{|c|}
    \multicolumn{1}{c}{Encoder} \\
	\hline 
        Input  \\ 
        \hline 
        Flatten  \\
        \hline 
        Linear, BatchNorm2D, l-ReLU \\
        \hline 
        Linear, BatchNorm2D, l-ReLU \\
        \hline 
        Linear, BatchNorm2D, l-ReLU \\
        \hline 
        Linear, BatchNorm2D, l-ReLU \\
        \hline 
        Linear, BatchNorm2D, Tanh \\
        \hline
        Quantizer \\
    \hline
    \end{tabular}
    \quad
    \begin{tabular}[t]{ |c| }
    \multicolumn{1}{c}{Decoder} \\
	\hline 
        Input  \\
        \hline 
        Linear, BatchNorm1D, l-ReLU \\
        \hline 
        Linear, BatchNorm1D, l-ReLU \\
        \hline
        Unflatten \\ 
        \hline 
        ConvT2D, BatchNorm2D, l-ReLU \\
        \hline 
        ConvT2D, BatchNorm2D, l-ReLU \\
        \hline
        ConvT2D, BatchNorm2D, Sigmoid \\
    \hline
    \end{tabular}
    \quad
    \begin{tabular}[t]{ |c| }
    \multicolumn{1}{c}{Discriminator} \\
	\hline 
        Input  \\ 
        \hline 
        Conv2D, l-ReLU \\
        \hline 
        Conv2D, l-ReLU \\
        \hline
        Conv2D, l-ReLU \\
        \hline
        Linear  \\  
    \hline
    \end{tabular}
    \quad
    \begin{tabular}[t]{ |c| }
    \multicolumn{1}{c}{Classifier} \\
	\hline 
        Input  \\
        \hline 
        Conv2D (10 filters, kernel=5), ReLU \\
        \hline 
        MaxPool2D (kernel size=2) \\
        \hline
        Conv2D (10 filters, kernel=5), ReLU \\
        \hline
        MaxPool2D (kernel size=2) \\
        \hline
        Flatten \\
        \hline
        Linear, ReLU \\
        \hline
        Linear, Softmax \\
    \hline
    \end{tabular}
\end{table}

Table \ref{tab:hyperparameters} summarizes the hyperparameters used in training each component of the model across all experiments.

\begin{table}[!htb]
    \centering
    \caption{Training hyperparameters used across all experiments.}
    \label{tab:hyperparameters}
    \begin{tabular}{ccccc}
    \specialrule{.1em}{.05em}{.05em} 
    & $\alpha$ & $\beta_1$ & $\beta_2$ & $\lambda_{\text{GP}}$ \\ 
    \hline
    Encoder & $10^{-2}$ & $0.5$ & $0.9$ & - \\
    Decoder & $10^{-2}$ & $0.5$ & $0.9$ & - \\
    Critic & $2 \times 10^{-4}$ & $0.5$ & $0.9$ & $10$ \\
    Classifier & $10^{-3}$ & $0.9$ & $0.999$ & - \\
    \hline
    \end{tabular}
\end{table}

\end{document}